\documentclass{llncs}
\usepackage[german,english]{babel}
\usepackage{times,theorem}
\usepackage{latexsym,color,graphics}

\usepackage{epsfig}
\usepackage{amssymb}
\usepackage{amsmath}
\usepackage{amsfonts}
\usepackage{xspace}
\usepackage{graphicx}
\begingroup
\catcode`\~=11
\gdef\urltilde{\lower 0.6ex\hbox{~}}
\endgroup

\newcommand{\A}{\mathcal{A}} \newcommand{\B}{\mathcal{B}}
 \newcommand{\D}{\mathcal{D}}

\newcommand{\I}{\mathcal{I}} 
\newcommand{\K}{\mathcal{K}} \renewcommand{\L}{\mathcal{L}}
 \newcommand{\N}{\mathcal{N}}
 \renewcommand{\P}{\mathcal{P}}
 
 \newcommand{\T}{\mathcal{T}}
 \newcommand{\V}{\mathcal{V}}
\newcommand{\W}{\mathcal{W}}

\title{Neuro-Symbolic Strong-AI Robots  with  Closed Knowledge Assumption: Learning and Deductions}
\author{Zoran Majki\'c}
\authorrunning{Zoran Majki\'c}
\institute{ISRST, Tallahassee, FL, USA\\
\email{majk.1234@yahoo.com}}

\newtheorem{coro}{Corollary}

\begin{document}
\maketitle              

\begin{abstract}
Knowledge representation formalisms are aimed to represent general conceptual information and are typically used in the construction of the knowledge base of reasoning agent. A knowledge base can be thought of as representing the beliefs of such an agent. Like a child, a strong-AI (AGI) robot would have to learn through input and experiences, constantly progressing and advancing its abilities over time. Both with statistical AI generated by neural networks we need also the concept of \textsl{causality} of events traduced into directionality of logic entailments and deductions in order to give to robots the emulation of human intelligence. Moreover, by using the axioms we can guarantee the \textsl{controlled security} about robot's actions based on logic inferences.\\
  For AGI  robots we consider the 4-valued Belnap's bilattice of truth-values  with knowledge ordering as well, where the value "unknown" is the bottom value, the sentences with this value are indeed unknown facts, that is, the missed knowledge in the AGI robots. Thus, these unknown facts are not part of the robot's knowledge database, and by learn through input and experiences, the robot's knowledge would be naturally expanded over time. \\
  Consequently, this phenomena can be represented by the  Closed Knowledge Assumption and Logic Inference provided by this paper.
  Moreover, the truth-value "inconsistent", which is the top value in the knowledge ordering of Belnap's bilattice, is necessary for strong-AI robots to be able to support such inconsistent information and paradoxes, like Liar paradox, during deduction processes.
 \end{abstract}

\section{Belnap's 4-valued Logic for AGI Robots}
Many-valued logic was conceived as a logic for uncertain, incomplete
and possibly inconsistent information which is very close to the
statements containing the words "necessary" and "possible", that is,
to the statements that make an assertion about the \emph{mode of
truth} of some other statement. \emph{Algebraic} semantics interprets modal connectives as operators, while \emph{Relational} semantics uses
relational structures, often called Kripke models, whose elements
are thought of variously as being possible worlds; for example,
moments of time, belief situations, states of a computer, etc.. The
two approaches are closely related: the subsets of relational
structures form an algebra with modal operators, while conversely
any modal algebra can be embedded into an algebra of subsets of a
relational structure via extensions of Stone's
representation theory.

Many classical as well as non-classical techniques for modeling the reasoning of intelligent systems have been proposed. Several frameworks, based on a many-valued logic with truth/knowledge partial orders (and lattices), for manipulating uncertainty, incompleteness, and inconsistency have been proposed in the form of particular extensions of classical logic programming and deductive databases. Such truth/knowledge ordered lattices can be integrated into a unique structure called bilattice \cite{Beln77,Gins88,Fitt91}, and Ginsberg has shown that the same theorem prover can be used to simulate reasoning in first order logic, default logic, prioritized default logic and assumption truth maintenance system.

Bilattice theory is a ramification of multi-valued logic by
considering both truth $\leq$ and knowledge $\leq_k$  partial
orderings. Given two truth-values $x,y \in \B$, if $x \leq y$ then
$y$ is at least as true as $x$, i.e., $x \leq y$ iff $x < y$ or
$x = y$. The two operations corresponding to this ordering
(t-lattice) are the meet (greatest lower bound) $\wedge$ and the
join (least upper bound) $\vee$.
For the knowledge-ordering $x \leq_k y$ means that $y$ is more precise than $x$, i.e.,
$x \leq_k y$ iff $x <_k y$ or $x = y$ . The  operations $\otimes$
and $\oplus$ correspond to the greatest lower bound  and least upper
bound respectively in the knowledge-ordering (k-lattice). The
negation operation for these two orderings are defined as the
involution operators which satisfy De Morgan law between the join
and meet operations.
\begin{definition} (Ginsberg ~\cite{Gins88})  \label{def:billat} A bilattice $\B$ is
defined as a sextuple $(\B, \wedge, \vee, \otimes, \oplus, \neg)$,
such that:\\
1. The t-lattice $(\B, \leq, \wedge, \vee)$ and the k-lattice
$(\B, \leq_k,\otimes, \oplus)$ are both complete lattices.\\
2. $\neg:\B \rightarrow \B$ is an involution ($\neg \neg$ is the
identity) mapping such that:\\
$\neg$ is lattice homomorphism from $(\B, \wedge, \vee)$ to $(\B,
\vee, \wedge)$ and $(\B,\otimes, \oplus)$ to itself.
\end{definition}
The two partial orders $\leq$ and $\leq_k$ represent respectively how much confidence we have in the validity (truth) of a particular sentence, and how much information (knowledge) we have about it.
Notice that from this definition, the negation $\neg$ is an
antitonic operator w.r.t. the $\leq$, with $\neg 1_t = 0_t$, $\neg
0_t = 1_t$ (where $0_t, 1_t$ are the bottom and the top elements
w.r.t the truth-ordering $\leq$ respectively), but \emph{monotonic} w.r.t. the knowledge ordering $\leq_k$. This homomorphism $\neg$ inverts the truth partial
ordering, i.e., if $x\leq y$ then $\neg x \geq \neg y$, while
it preserves the knowledge preordering, i.e., if $x\leq_k y$ then
$\neg x \leq_k \neg y$.\\
A valuation of a set of sentences $\L_0$ is denoted by mapping $v^*:\L_0\rightarrow \B$. They generate a bilattice as well with partial orders $\preccurlyeq$ and $\preccurlyeq_k$ respectively:
\begin{definition} (Ginsberg ~\cite{Gins90})  \label{def:billat2} Given the valuations $v^*_1$ and $v^*_2$, we will say that $v^*_2$ is a \textsl{knowledge extension} of $v^*_1$, writing $v^*_1\preccurlyeq_k v^*_2$ if $v^*_1(\phi) \preccurlyeq_k v^*_2(\phi)$ for every sentence in $\L_0$ (analogously we have for truth extension $\preccurlyeq$ as well).\\
A valuation $v^*$ will be called \textsl{closed} if:\\
1. $v^*(\wedge_i \phi_i) \geq_k \wedge_iv^*(\phi_i)$.\\
2. (the negation of logic syntax is that of bilattice): $v^*(\neg\phi) = \neg v^*(\phi)$.\\
3. (truth preservation of entailment $\models_{v^*}$):  if $\phi\models_{v^*} \psi$, then $v^*(\phi)\leq v^*(\psi)$.
\end{definition}
 The smallest \emph{nontrivial} bilattice is Belnap's 4-valued
bilattice  ~\cite{Beln77},
 \begin{figure}
$\vspace*{-9mm}$
\centering{
 \includegraphics[scale= 1.11]{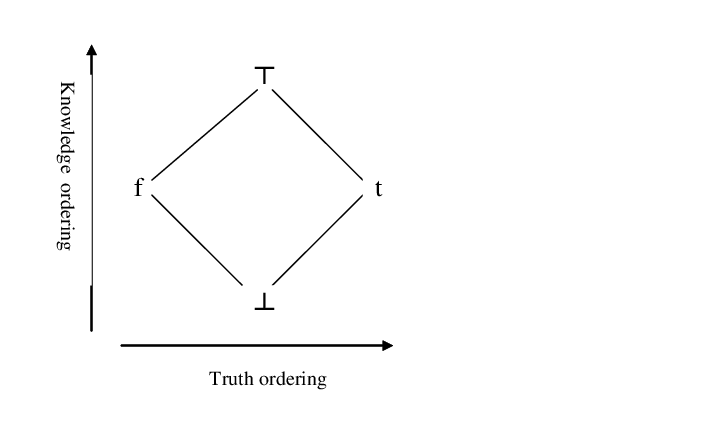}
  \caption{Belnap's bilattice}
  \label{fig:Operations}
 }
  $\vspace*{-11mm}$
 \end{figure}
 $X= \B_4 = \{t,f,\bot, \top\}$  in Fig.\ref{fig:Operations},
 with four truth-values in $X$ where $t $ is \emph{true}, $f $ is \emph{false}, $\top$ is inconsistent (both true and false) or \emph{possible }, and $\bot$ is \emph{unknown} truth-value.
In what follows  we denote by $x \bowtie y$ two unrelated elements in $X$ (so that not $(x\leq y$ or $y\leq x)$). So, Belnap's bilattice is composed by the truth lattice $(X,\leq)$ and knowledge lattice $(X,\leq_k)$,
with two natural orders:
\emph{truth} order, $\leq$, and \emph{knowledge} order, $\leq_k$,
such that $f \leq \top \leq t$, $~f \leq \bot \leq t$, $\bot
\bowtie_t \top$ and $\bot \leq_k f \leq_k \top$, $~\bot \leq_k t
\leq_k \top$, $f \bowtie_k t$. That is, bottom element for
$\leq$ ordering is $f$, and for $\leq_k$ ordering is $\bot$,
 and top element for $\leq$ ordering is $t$, and for $\leq_k$ ordering is $\top$. Meet and join operators under $\leq$ are denoted $\wedge$ and
$\vee$; they are natural generalizations of the usual conjunction
and disjunction notions. Meet and join under $\leq_k$ are denoted
$\otimes$ and $\oplus$, such that hold: $~f \otimes t = \bot$, $f
\oplus t =\top$, $\top\wedge \bot = f$ and $\top \vee \bot = t$.

In \cite{Majk07MV} we proposed a family of intuitionistic bilattices with full
truth-knowledge duality (D-bilattices) for a logic programming: the simplest of them, based on intuitionistic truth-functionally complete extension of Belnap's
4-valued bilattice, can be used in paraconsistent programming, that
is, for knowledge bases with incomplete and inconsistent
information.

There is a close relationship between logic programming and
inductive definitions which are forms of \emph{constructive}
knowledge \cite{FBry89,DeBM01}. Constructive information defines a
collection of facts through a constructive process of iterating a
recursive recipe. This recipe defines new instances of this
collection in terms of presence (and sometimes the absence) of other
facts of the collection. In the context of mathematics, constructive
information appears by excellence in inductive definitions. Not a
coincidence, inductive definitions have been studied in constructive
mathematics and intuitionistic logic, in particular in the sub-areas
of Inductive and Definition logics,
Iterated Inductive Definition logics and Fixpoint logics. \\
So, we adopt the idea of using relative pseudo
complements as implications, which has been proposed in the
literature several times and is nowadays well understood, for
distributive bilattices also. In fact, the use of relative
pseudo-complements is adopted in the fuzzy, multi-adjoint and
residuated logic programming and allows also the assignment of
weights to rules. Consequently,  here one
particular family of \emph{intuitionistc}
 bilattices, that is, distributive bilattices extended by  logic implication $ \Rightarrow $, defined as a relative pseudo complement
  $x \Rightarrow y  =\bigvee \{ z | z \wedge x \leq y \}$ (here  $\wedge$ is the meet operation, $\leq$ is the truth-ordering of a bilattice and $\bigvee$ is the l.u.b). It
has nice mathematical properties \cite{MajkA04,MajkC04,Majk05f}.
 Semantically, the basic notion is that of relative
pseudo complement from the algebraic semantics of intuitionistic
logic \cite{RaSi70}. More over, such an implication can be used for
\emph{nested} implications in bodies of rules, in order to extend
the expressive power of logic programming (see, for example, in
\cite{Mill89}). For Belnap's bilattice we obtain the following table for such logic implication:
\begin{equation} \label{tab:impl}
\begin{tabular}{|c|c|c|c|c|}
  \hline
  $\Rightarrow$ & $~~t~~~\bot ~~~f ~~~\top $   \\
  \hline
  $t $ & $~~t~~~\bot ~~~f ~~~\top $  \\
  $\bot$ & $~~t~~~~~t ~~~\top~~\top$\\
  $f$ & $~~t~~~~~t ~~~~t ~~~~t$ \\
    $\top$  & $~~t~~~\bot~~\bot~~ ~t$  \\
 \hline
\end{tabular}
\end{equation}
The truth lattice will be used for the many-valued many-sorted intensional FOL (denoted by $IFOL_B$) while the knowledge lattice can by used for estimation of the rate of growing the robots knowledge in time, as result of robot's experience and deductions. In order to simplify presentation we will use the same symbols of logic connectives of $IFOL_B$ as that of Belnap's bilattice for conjunction, disjunction, implication and negation.
So, the conjunction and disjunction of $IFOL_B$ are consequently the meet and join operators under $\leq$, denoted $\wedge$ and
$\vee$.  For the implication $\Rightarrow$ of  $IFOL_B$ we can use the \emph{relative pseudocomplements} defined above, while for the negation, the bilattice involution maping $\neg$ (which is a\emph{paraconsistent} negation that reverses the $\leq$ ordering, while preserving the $\leq_k$ ordering)\footnote{Differently from it, the negation used as pseudocomplement, $\neg_t x = x \Rightarrow f$, reverses both orderings.}: switching $f$ and $t$, leaving $\bot$ and $\top$.

The many-valued equivalence operator is defined for any two $a,b \in X$ by:
\begin{equation} \label{eq:dueM2}
a \Leftrightarrow b =  \left\{
    \begin{array}{ll}
      t, & \hbox{if $~a = b$}\\
      f, & \hbox{otherwise}
    \end{array}
  \right.
\end{equation}

 In what follows we will consider only conservative homomorphic
many-valued extensions of the classic 2-valued logic with logical
connectives $\neg, \wedge, \vee,  \Rightarrow, \Leftrightarrow$
(negation, conjunction, disjunction, implication and equivalence
respectively) and of the following algebra of truth values in the lattice $(X,\leq)$:
\begin{definition} \label{def:mv-algebra}
 For a given many-valued predicate logic,  by $\A_{B} = (X, \leq, \neg, \wedge, \vee,  \Rightarrow,\Leftrightarrow)$ we denote the algebra of the Belnap's truth values in the complete lattice $(X, \leq)$.\footnote{The classic 2-valued algebra $\A_{\textbf{2}}$ is particular case with lattice $\textbf{2} = \{f,t\}$ and if $\textbf{2} \subset X$  (the top and bottom elements in $X$).}
\end{definition}
The fact that classic 2-valued algebra is the subalgebra of this
many-valued algebra $X$, means that $X$ is a conservative
extension of the classic logic operators.
We denote by $\textbf{2} = \{f,t\}\subseteq X $ the
classical 2-valued logic lattice, which is the sublattice of the Belnap's truth-ordering lattice $(X, \leq)$ composed only by false and true truth-values $\{f,t\}$. It is easy to verify that Definition \ref{def:HerbrandValuat} is valid also for classical FOL and its 2-valued logic lattice $\textbf{2}$. However, the Belnap's bilattice has also the knowledge ordering which is very useful for estimation of the robot's knowledge, because for a robot, as default any ground atom has the unknown truth-value $\bot$.\\
\textbf{Remark}: It has been demonstrated by Lemma 1 in  \cite{Majk07MV} that the unary operator of negation $\neg$ in the Belnap's algebra $\A_B$ is the selfadjoint modal operator w.r.t the knowledge-ordering $\leq_k$, and antitonic auto-homomorphism of the truth-ordering lattice, $\neg:(X,\leq,\wedge,\vee) \rightarrow (X,\leq,\wedge,\vee)$, which represents the De Morgan laws.
 So, the algebra of truth values $\A_B$ is a \emph{modal} Heyting algebra (that is, a standard Heyting algebra extended by unary modal operator).\\
$\square$\\
The central hypothesis of cognitive science is that thinking can best be understood in terms of representational structures in the mind and computational procedures that operate on those structures. Most work in cognitive science assumes that the mind has mental representations analogous to computer data structures, and computational procedures similar to computational algorithms.

Neuro-symbolic AI attempts \cite{Vali08,BaHi05,SZEH21,PanJ23,CoRe24} to integrate neural and symbolic architectures in a manner that addresses strengths and weaknesses of each, in a complementary fashion, in order to support robust strong AI capable of reasoning, learning, and cognitive modeling. In this paper we consider the
 Intensional First Order Logic (IFOL)  as a symbolic architecture of modern robots, able to use natural languages to communicate with humans and to reason about their own knowledge with self-reference and abstraction language property.
After publication of my book \cite{Majk22} in 2022 for such IFOL, recently I applied the standard 2-valued version of it \cite{Majk23r,Majk24a,Majk24ar} to the four-level cognitive structure of AGI robots here adopted for Belnap's 4-valued bilattice.

It requires the conceptual PRP structure levels based on domain $\D$ composed by particulars in $D_{0}$ and universals (concepts) in $D_I = D_1+D_2+D_3+...$
 The universals in $D_I$ is made of: $D_1$ for many-valued logic sentences, as described in Section 2 of \cite{Majk24ar}, called \verb"L-concepts" (their extension corresponds to some logic value),   and  $D_n, n \geq 2,$ for
 \verb"concepts" (their m-extension is an n-ary relation); we consider the property (for an unary predicate) as a  concept in $D_2$.

 We are able to incorporate the emotional structure to robots as well, by a number of fuzzy-emotional partial mappings
\begin{equation} \label{eq:emotions}
E_i:\D \rightarrow [0,1]
\end{equation}
 of robots PRP intensional concepts,  for each kind of emotions $i\geq 1$: love, beauty, fear, etc. It was demonstrated \cite{Majk22} that IFOL is able to include any kind of many-valued, probabilistic and fuzzy logics as well.

Despite the best efforts over the last years, deep learning is still easily fooled \cite{NYCl15}, that is, it remains very hard to make any guarantees about how the system will behave given data that departs from the training set statistics. Moreover, because deep learning does not learn causality, or generative models of hidden causes, it remains \emph{reactive}, bound by the data it was given to explore \cite{Marc18}.

In contrast, brains act proactively and are partially driven by endogenous curiosity, that is, an internal, epistemic, consistency, and knowledge-gain-oriented drive. We learn from our actively gathered sensorimotor experiences and form conceptual, loosely hierarchically structured, compositional generative predictive models.

The robot's internal neuro-symbolic knowledge structure  \cite{Majk23r,Majk24a} is divided into four levels, in ordering: natural language $\N\L_{list}$ (parsed into conceptual structures with PRP concepts, i.e., conceptual C-structures\footnote{Following the development of advanced neural network techniques, especially the Seq2Seq model, and the availability of powerful computational resources, neural semantic parsing started emerging. Deep semantic parsing attempts to parse natural language utterances, typically by converting them to a formal meaning representation language which in our case are the expressions of the algebra of PRP concepts so defining the intensional (partial) mapping $I:\L \rightarrow D_I$ (represented completely by Frege/Russel semantics diagram of Corollary 3 in \cite{Majk25} as well. Good candidate for conceptual C-structures are AMR (Abstract Meaning Representation, and its improved models, with AMR graphs which are rooted, labeled, directed, acyclic graphs (DAGs), comprising whole sentences, by mapping their nodes and leaves into PRP intensional concepts and particulars of $IFOL_B$.} and  formal IFOL logic  with the set of formulae $\L$),  intensional conceptual PRP system and neuro system, as represented in the figure in Appendix (Section 7).

That means, for example, that we can use neural networks for robot's transformation of the spoken phrases or  visual contents presented to robots into written language phrases. For example, we can use  the end-to-end learning paradigm and design neural network architectures for the tasks of image-sentence matching, image captioning tasks and image region-annotation as indicating by the PhD thesis of Andrej Karpathy in 2016 \cite{Karp16} with novel techniques and algorithms in the neural networks modeling paradigm which offers superior results and multiple appealing practical properties by introducing the dense captioning task, which requires
the computer to both detect and describe all salient regions of an image.

To be grounded, the symbol system of $IFOL_B$ would have to be augmented with nonsymbolic, sensorimotor capacities—the capacity to interact autonomously with that world of objects, events, actions, properties and states that their symbols are systematically interpretable (by us) as referring to. It would have to be able to pick out the referents of its symbols, and its sensorimotor interactions with the world would have to fit coherently with the symbols' interpretations.

Note that, while the top line in the figure in Appendix (Section 7) is the ordinary component of the natural language grounding developed by robot's neuro system, the two lines bellow is the new robots knowledge structure of the added \emph{symbolic AI system} based on the Intensional First Order Logic and its grounding to robot's processes (its neuro AI system), by which the robot is able to provide logic deductive operations  and autoepistemic self-reasoning about its current knowledge states and communicate it to humans by using natural languages.

Thus, the robot-s cognitive system in the figure  in Appendix (section 7)combines both approaches to the linguistics: the \emph{rationalist}  Chomskyan (or
generative) linguistics approach which depends on \emph{categorical} principles, which sentences either do or do not satisfy,  and the \emph{empiricist} approach (of the British linguist J.R. Firth, for example) which is in favor to the \emph{statistical}  linguistics (probabilistic, actually implemented in LLM by neural networks, for example), better at automatic learning (knowledge induction), better at disambiguation, and also have a role in the science of linguistics.
The rationalist approach  postulates that the key parts of language are innate-hardwired in the brain at birth as part of the human genetic inheritance, mathematically modeled by formal symbolic logics, the knowledge databases and reasoning systems (knowledge deduction). Differently, statistical linguistics draws from the work of Shannon, where the aim is to assign probabilities to linguistic events, so that we can say which sentences are usual and unusual. With this approach we can learn the complicated and extensive structure of language
by specifying an appropriate general language model, and then inducing
the values of parameters by applying statistical, pattern recognition, and
machine learning methods to a large amount of language use.

Next Section is dedicated to a short description of the Many-sorted Intensional First-order logic $IFOL_B$ based on the Belnap's 4-valued bilattice of truth-values with the set of variables $\V$ and domain $\D$, provided recently in \cite{Majk25}).
\section{Introduction to 4-valued Many-sorted Intensional FOL}
The significant aspect of an expression's meaning is its \emph{extension}.
We can stipulate that the extension of a sentence is its
truth-value, and that the extension of a singular term is its
referent. The extension of other expressions can be seen as
associated entities that contribute to the truth-value of a sentence
in a manner broadly analogous to the way in which the referent of a
singular term contributes to the truth-value of a sentence.

The first conception of \emph{intensional entities} (or concepts) is built
into the \emph{possible-worlds} treatment of Properties, Relations
and Propositions (PRP)s. This conception is commonly attributed to
Leibniz, and underlies Alonzo Church's alternative formulation of
Frege's theory of senses ("A formulation of the Logic of Sense and
Denotation" in Henle, Kallen, and Langer, 3-24, and "Outline of a
Revised Formulation of the Logic of Sense and Denotation" in two
parts, Nous,VII (1973), 24-33, and VIII,(1974),135-156). This
conception of PRPs is ideally suited for treating the
\emph{modalities} (necessity, possibility, etc..) and to Montague's
definition of intension of a given virtual predicate
$\phi(x_1,...,x_k)$, as a mapping from possible worlds into
extensions of this virtual predicate. Among the possible worlds we
distinguish the \emph{actual} possible world\footnote{For example, if we
consider a set of predicates, of a given Database,
and their extensions in different time-instances, then the actual possible world is identified by the current instance of the time.}.

The second conception of intensional entities is to be found in
Russell's doctrine of logical atomism. In this doctrine it is
required that all complete definitions of intensional entities be
finite as well as unique and non-circular: it offers an
\emph{algebraic} way for definition of complex intensional entities
from simple (atomic) entities (i.e., algebra of concepts),
conception also evident in Leibniz's remarks. In a predicate logics,
predicates and open-sentences (with free variables) expresses
classes (properties and relations), and sentences express
propositions. Note that classes (intensional entities) are
\emph{reified}, i.e., they belong to the same domain as individual
objects (particulars). This endows the intensional logics with a
great deal of uniformity, making it possible to manipulate classes
and individual objects in the same language.

 In what follows any open-sentence, a formula $\phi(\textbf{x})$ with non empty tuple of free variables $\textbf{x} =(x_1,...,x_m)$, will be called a m-ary   \emph{virtual predicate}\footnote{In intensional FOL we use virtual pradicates so that their extension for a given interpretation generates a single m-ary relation as in the case of standard m-ary predicates used in FOL for atomic formulae with $m\geq 1$ free variables. Each virtual predicate generates a m-ary \emph{intensional concept} in domain $\D$ for a given intensional interpretation as standard m-ary predicates.  Most simple syntactically virtual predicates are obtained from atoms, obtained from k-ary predicate letters $p^k_i \in P$, with non empty set of variables and non empty set of constants for predicate arguments as, for example $p^6_i(t_1,t_2,x_1,t_3,x_2,t_4)$ with ground terms (or constants), $t_m$ for $1\leq m\leq 4$, generates a virtual predicate $\phi_i(x_1,x_2)$. Virtual predicate of an atom with all its arguments represented by variables is just equal to such atom.}, denoted also by $\phi(x_1,...,x_m)$. This definition contains the precise method of establishing the \emph{ordering} of variables in this tuple:
\begin{definition} \label{def:virt-predicate} \textsc{Virtual predicates:}
\emph{Virtual predicate} obtained from an open formula $\phi \in \L$
is denoted by $\phi(x_1,...,x_m)$ where $(x_1,...,x_m)$ is a particular fixed sequence of the set of all free variables in $\phi$. This definition contains the precise method of establishing the \emph{ordering} of variables in this tuple:
such an method that will be adopted here is the ordering of appearance, from left to right, of free variables in $\phi$. This method of composing the tuple of free variables is  unique and canonical way of definition of the virtual predicate from a given open formula.
\end{definition}
Virtual predicates are used to build the \emph{semantic logic structures} of logic-semantics level of any given natural language. However, with virtual predicates we need to replace the general FOL quantifier on variables $(\exists x)$ by specific existential quantifiers $\exists_i$ of the intensional logic $IFOL_B$, where $i\geq 1$ is the position of variable $x$ inside a virtual predicate. For example, the intensional FOL formula $(\exists x_k) \phi(x_i,x_j,x_k,x_l,x_m)$ will be mapped into intensional concept $\exists_3 \phi(\textbf{x})$ where $\textbf{x}$ is the list(tuple) of variables $(x_i,x_j,x_k,x_l,x_m)$.
In the same way we introduce the set of universal quantifiers $\forall_i$.  Notice that in many-valued version of intensional FOL the composed operation $\neg \exists_i \neg$ \emph{is different} from $\forall_i$. Moreover with virtual predicates we need to replace standard binary logic connectives $\wedge, \vee, \Rightarrow$ with $\wedge_S, \vee_S, \Rightarrow_S$ where $S$ is a set of pair of indices of variables in virtual predicate. So we obtain a non standard FOL syntax algebra with such expressions that use these new connectives.

We assume, as in \cite{Majk22},  that a concept algebra of many-valued logic $IFOL_B$ has a non empty domain (it is different from the domain of 2-valued standard IFOL) $~\D = D_{0} + D_I$, (here $+$ is a disjoint union) where a subdomain $D_{0}$ is made of  particulars or individuals (we denote by $\circledR$ the non-meaning individuals, for interpretation of language entities that are no meaningful, as "Unicorn" for example) with $X \subseteq D_0$.

The rest $D_I = D_1 + D_2 ...+ D_n ...$ is made of
 universals (concepts): $D_1$ for many-valued logic sentences, as described in Section 2 of \cite{Majk24ar}, called \verb"L-concepts" (their extension corresponds to some logic value),   and  $D_n, n \geq 2,$ for
 \verb"concepts" (their m-extension is an n-ary relation); we consider the property (for an unary predicate) as a  concept in $D_2$. The  concepts in $\D_I$ are denoted by $u,v,...$, while the  values (individuals) in $D_0$ by $a,b,...$

 \begin{definition} \label{def:m-sortedExtensions} \textsc{m-sorted extensions}:\\
 We define  the set of all m-extensions in the many-sorted framework of universals in $D_I$,
 \begin{equation} \label{eq:dueM5}
 \mathfrak{Rm} = \widetilde{X}\bigcup\{ R \in
~\bigcup_{n \geq 1}\P(\D^n \times
(X\backslash\{\bot\}))~|~(u_1,..., u_{n}, a), ~(u_1,..., u_{n}, b) \in R ~~implies ~~b = a \}
\end{equation}
 so that each $(n+1)$-ary relation is a graph of a function, and by  $~\mathfrak{Rm}_k, ~k \geq 1$, we will denote the subset of all k-ary relations in $\mathfrak{Rm}$, and by $\emptyset$ each empty relation in $\mathfrak{Rm}$,
with the set of unary relations in $\widetilde{X}$,
\begin{equation} \label{eq:dueM6}
\widetilde{X} =
\{\{ a\}~|~ a \in X ~~and ~~a \neq \bot \} \bigcup \{\emptyset\} =  \{ \emptyset, \{ f\}, \{ \top\},\{ t\}\}
\end{equation}
where $\{ a\}$ denotes unary relation with unique tuple equal to the truth-value $a `in X$ and $\emptyset$ unary empty relation,
and hence, $\emptyset\in \widetilde{X}\subset \P(X)$.  So, $(\widetilde{X},\preceq)$ is complete lattice with total ordering $\emptyset \preceq \{f\}\preceq \{\top\}\preceq \{t\}$, that is, with the  bottom element and   $\{a\} \preceq \{b\}$ iff $a\leq b$ (w.r.t the \emph{truth-ordering} in $X$)\footnote{However, if we consider the inconsistent truth-value $\top$ as both true and false, that is, as set $\{f,t\}$ and unknown truth-value as empty set $\emptyset$, then $\widetilde{X} =
 \{ \emptyset, \{ f\}, \{ f,t\},\{ t\}\}$ is the powerset $\widetilde{X}= \P(\{ f,t\})) =\P(\textbf{2})$, we obtain the following isomorphism with \emph{knowledge-ordering} lattice $(X,\leq_k)$ of Bedlnap's bilattice and the complete powerset lattice $(\widetilde{X}, \subseteq)$,
 $$is:(X,\leq_k)\simeq (\widetilde{X}, \subseteq),$$
 such that $is(\bot) = \emptyset$, $is(f) = \{f\}$, $is(t) = \{t\}$ and $is(\top) = \{f,t\}$, from which we see the natural correspondence of $\widetilde{X}$ with knowledge ordering and hence with knowledge extensions.
 }.
\end{definition}
Note that $\mathfrak{Rm}$ contains only \emph{known} extensions of the many-valued concepts, and this explains also why we defined for $\bot \in X$ the empty set $\emptyset$.

 We define the following partial ordering $\preceq$ for
the m-extensions in
$\mathfrak{Rm}$ (note that for a given k-ary relation with $k\geq 2$ and $1\leq i\leq k$, the $\pi_i(R)$ is the i-th projection of $R$ while $\pi_{-1}(R)$ a relation obtained from $R$ by eliminating i-th column of it):
\begin{definition} \label{def:POrder} \textsc{Extensional partial order in} $\mathfrak{Rm}$:\\
We extend the ordering of unary relations in $\widetilde{X }\subset\mathfrak{Rm}$ by:
for any two nonempty m-relations $R_1, R_2 \in \mathfrak{Rm}$ with arity $k_1 = ar(R_1), k_2 = ar(R_2)$ and $m = \max(k_1,k_2)\geq 2$, $\emptyset \preceq R_i$ for $i =1,2$, and

$~~~R_1\preceq R_2~~$ iff $~~$ \\for each $(u_1,...,u_{m-1},a) \in
Ex(R_1,m)$,  $\exists(u_1,...,u_{m-1},b) \in Ex(R_2,m)$
with $a \leq b$,\\
where this expansion-mapping $Ex:\mathfrak{Rm}\rightarrow
\mathfrak{Rm}$ is defined as follows for any nonempty $R \in \mathfrak{Rm}$:
\begin{equation} \label{eq:dueM7}
Ex(R,m) =
 \left\{
    \begin{array}{ll}
      \D^{m-k}\times
\{a\}, ~~for ~~R = \{a\}, ~~a \in X\backslash\{\bot\} & \hbox{if $~~m > k=1$}\\
  \bigcup \{\{(u_1,...,u_k)\}\times \D^{m-k}\times
\{a\}~|~(u_1,...,u_k,a) \in R\}, & \hbox{if $~~m > k> 1$}\\
      R, & \hbox{otherwise}
    \end{array}
  \right.
\end{equation}
We denote by $R_1 \simeq R_2$ iff $R_1\preceq R_2$ and $R_2\preceq
R_1$, the equivalence relation between m-extensions.
\end{definition}
We chose $\emptyset, \{t\} \in  \widetilde{X} \subset \mathfrak{Rm}$ to be the representative elements of the bottom and top equivalence classes in $\mathfrak{Rm}$ relatively (determined by $\simeq$), so that $(\forall R \in \mathfrak{Rm})( \emptyset \preceq R \preceq \{t\})$. Consequently, the \emph{extensional} ordering $\preceq$ in $\mathfrak{Rm}$ extends the ordering of the lattice  $(\widetilde{X},\preceq)$.

We introduce for the concepts in $D_I$ an \emph{extensional interpretation}  $h$ which assigns the \emph{m-extension} to each intensional concept in $\D$,
and can be considered as an \emph{interpretation of concepts} in
$\D$. Thus, each concept in $D_I$ represents a set of tuples in $\D$,
and can be also an element of the extension of another concept and
of itself also.  Each extensional interpretation $h$ assigns to the intensional elements of $\D$
an appropriate extension: in the case of particulars $u \in  D_{0}$, $h_0(u) \in D_0$, such that for each logic value $a\in X\subset D_0$, $h_0(a) = a$.  Thus, we have the particular's mapping $h_0:D_0\rightarrow D_0$ and more generally  (here $'+'$ is considered as disjoint union),
 \begin{equation}\label{eq:dueM6}
 h =   \sum_{i\in \mathbb{N}}h_i:\D \longrightarrow   D_0+\sum_{i\geq 1}\mathfrak{Rm}_i
\end{equation}
 where $h_1:D_1 \rightarrow \widetilde{X}$ assigns to each  L-concept $u \in D_1$  (of a sentence with truth-value $a\in X$), a
relation composed by the single tuple $h_1(u) = \{a\}$ if $a \neq \bot$),  $~\emptyset$ otherwise, and $h_i:D_i\rightarrow \mathfrak{Rm}_i$, for $i\geq 2$, that assigns a m-extension to non-sentence concepts. The extension of the \emph{self-reference}
truth-concept $u_T\in D_2$, for any extensionalization  function $h$ is given by $h(u_T) = \{(u,a)~|~u \in D_1,
\emptyset \neq h(u) =\{a\}\}$. This concept represents the truth of
all logic-concepts in $D_1$, in the way that for any L-concept $u \in D_1$ we have that $~h(u) =\{a\}~$ iff $~(u,a) \in h(u_T)$.

In what follows, this logic $IFOL_B$  will be shortly denoted by $\L_{in}$ and the set of all its formulae by $\L$.
 \begin{definition} \label{def:syntax} \textsc{Syntax of many-valued intensional FOL $\L_{in}$}:\\
 We define the syntax of 4-valued many-sorted first-order logic $\L_{in}$,  with complete Belnap's lattice of truth-values $(X,\leq)$, by:\\
 - Variables $x,y,z,..$ in $\V$;\\
 - Language constants $c,d,...$ are considered as nullary functional letters (all non nullary functional letters are represented as the predicate
  letters for graphs of these functions);  \\
  - Predicate letters in $P$, denoted by $p_1^{k_1},p_2^{k_2},...$
   with a given arity $k_i\geq 1$,    $i = 1,2,..$  Nullary predicate letters are considered as propositional letters, as for example the built-in predicate letters $\{p_a~|~ a \in X\}$;\\
  - The 4-valued logic connectives:
   1. unary connectives are the negation $\neg$, and existential an universal quantifiers $\exists_i, \forall_i$.\\
   2. binary connectives in $\circledcirc \in \{\wedge,\vee,\Rightarrow,\Leftrightarrow\}$ are   respectively logic conjunction, disjunction,  implication and  equivalence.
   From the fact that we are using the virtual predicates that represent any formula composed by these connectives as a specific virtual predicate with free variables obtained as the union of the free variables of all its standard predicates in $P$ that compose this formula, for the syntax algebra $\mathcal{A}\mathfrak{B}_{FOL}$ 
   of $\L_{in}$ we need to use for each binary logic operator in $\circledcirc$ (different from $\Leftrightarrow$) an algebraic operator $\circledcirc_S$ where the pairs of indexes in the set $S$ indicate the equal free variables used in two subformulae $\psi_1$ and $\psi_2$ that compose a resulting virtual predicate $\phi(x)$ equal to formula $\psi_1\circledcirc_S \psi_2$ as it is explained\footnote{
  For example, the FOL formula
$\psi_1(x_i,x_j,x_k,x_l,x_m) \wedge \psi_2 (x_l,y_i,x_j,y_j)$ will be
replaced by a resulting \emph{virtual predicate} $\phi(x_i,x_j,x_k,x_l,x_m,y_i,y_j)$ defined by algebraic expression $\psi_1(x_i,x_j,x_k,x_l,x_m)
\wedge_S \psi_2~ (x_l,y_i,x_j,y_j)$, with $S = \{(4,1),(2,3)\}$, and then traduced by
the algebraic expression $~R_1 \bowtie_{S}R_2$  (see point 2.5 of Definition 8 in \cite{Majk25}) where $R_1 \in
\P(\D^5), R_2\in \P(\D^4)$ are the extensions for a given many-valued
interpretation $v^*$ of the virtual predicate $\psi_1, \psi_2$
relatively. In this example the resulting relation
will have the following ordering of attributes of obtained virtual predicate $\phi$:
$(x_i,x_j,x_k,x_l,x_m,y_i,y_j)$.\\
In the case when $S$ is empty (i.e. its cardinality $|S| =0$) then the resulting relation is the Cartesian product of $R_1$ and $R_2$.}
   in \cite{Majk25}.\\
  - Abstraction operator $\lessdot \_~\gtrdot$, and punctuation
 symbols (comma, parenthesis).\\
 With the following simultaneous inductive definition of \emph{terms} and \emph{formulae}:
 \begin{enumerate}
   \item All variables and constants (0-ary functional letters in P) are terms.
   We denote by $p_a$ a logic constant (built-in 0-ary predicate symbol) for each truth value $a \in X$, so that a set of constants is a not empty set.
   \item If $~t_1,...,t_k$ are terms, then $p_i^k(t_1,...,t_k)$ is a formula
 ($p_i^k \in P$ is a k-ary predicate letter).
   \item  In what follows any open-formula $\phi(\textbf{x})$ with non empty
tuple of free variables $\textbf{x} = (x_1,...,x_m)$, will be called a m-ary
  \emph{virtual predicate}, denoted also by
$\phi(x_1,...,x_m)$ and provided in Definition  \ref{def:virt-predicate}.\\
   If $\phi$ and $\psi$ are formulae, then   $\neg\phi$, $(\forall_k x) \phi$,  $(\exists_k x) \phi$ and, for each connective $\odot$, if binary $\phi \odot \psi$ or if unary $\odot\phi$, are the formulae. In a formula $(\exists_k )\phi(\textbf{x})$ (or $(\forall_k) \phi(\textbf{x})$), the virtual predicate $\phi(\textbf{x})$ is called "action field" for the quantifier $(\exists_k)$ (or $(\forall_k)$) of the k-th free variable in the tuple $\textbf{x}$. A variable $y$ in a formula $\phi$ is called bounded variable iff it is the variable quantified by $(\exists_k)$ (or $(\forall_k)$). A variable $x$ is
free in $\psi(\textbf{x})$ if it is not bounded. A \emph{sentence} is a closed-formula having no free variables.
\item If $\phi(\textbf{x})$   is a formula and $\alpha \subseteq \overline{\textbf{x}}$ is a possibly empty subset of  \emph{hidden} (compressed) variables, then $\lessdot \phi(\textbf{x}) \gtrdot_{\alpha}^{\beta}$ is an abstracted term,
    where $\beta$ is  remained subset of free visible variables in  $\phi$.  So, the subtuples of hidden and visible variables (preserving the ordering of the tuple $\textbf{x}$ are $\pi_{-\beta}\textbf{x}$ and $\pi_{-\alpha}\textbf{x}$, respectively).  If $\alpha$ or $\beta$ is empty sequence, than it can be
 omitted (for example, if $\phi$ is closed formula, then this term
 is denoted by $\lessdot \phi \gtrdot$).
  \end{enumerate}
An occurrence of a variable $x_i$ in  a term $\lessdot \phi(\textbf{x})
\gtrdot_{\alpha}^{\beta}$ is \emph{bound}  if $x_i \in \alpha$, free
if $x_i \in \beta$, so that the variables in $\alpha$ are not subjects of assignment $g\in \D^\V$ and can not be quantified by existential and universal FOL quantifiers.\\
 In particular we introduce this subset of distinguished predicates:
\begin{description}
 \item[a] The binary predicate letter $p_1^2 \in P$ is singled out as a distinguished logical predicate and formulae of the form $p_1^2(t_1,t_2)$ are to be rewritten in the form $t_1 = t_2$. It is a  built-in 2-valued\footnote{For "2-valued" we intend  the subset of standard truth-values $\textbf{2} =\{f,t\} \subset X$.} predicate for the identity.
  \item[b] The binary predicate letter $p_2^2 \in P$ is singled out as a
distinguished 2-valued  logical predicate and formulae of the form
$p_2^2(t_1,t_2)$ are to be rewritten in the form $t_1 =_{in} t_2$.
It is a built-in predicate for the weak-intensional-equivalence.
\item[c] The unary many/valued predicate $T$ and $Know$ for robot's own knowledge\footnote{In this paper we propose the minimal version of the predicate $Know \in P$ with first argument dedicated to time-specification of obtained knowledge, second dedicated for the subject of this knowledge and the last argument is abstracted term of the knowledge sentence: the truth-value of a ground atom of this predicate represents the truth-value of the known sentence contained in third argument.} management. The basic argument of them is an abstracted term obtained from a given sentence, so that the truth value of their ground atoms  reflects the truth/value of reified sentence in them. Consequently their interpretation has to satisfy the particular knowledge-constraints provided in (\ref{eq:esem3s}).
\end{description}
We denote by $\L$ the set of all formulae of $\L_{in}$ and by  $\L_0$ the subset of its sentences.
\end{definition}
Notice that "built-in" is used for predicate letters that have fixed
invariant interpretation, and consequently fixed invariant
extension: each possible interpretation of $\L_{in}$ has to satisfy
this constraint for built-in predicates, so that different many-valued
interpretations can be used only for the remaining set of predicate
letters in $P$. The $T$ is the many-valued version of the truth-predicate in the standard 2-valued logic where a formula $T(\lessdot\phi \gtrdot)$ is true iff the (closed) sentence $\phi$ is true.

 The intensional interpretation $I:\L \rightarrow D_I$ maps each nullary predicate symbol  into $D_1$, and for each logic constant $p_a, a \in X$, we have that $I(p_a) = u_a \in D_1$. The language constants (nullary functional symbols in $P$) are mapped in
concepts in $\D$. Notice that for any non-meaningful language constant (as "Unicorn" for example) we have that $I(c) = \circledR \in D_0$.

 Consequently, each atom $p_i^k(t_1,...,t_k)$, $p_i^k \in P$, is mapped into
the concept $u \in D_{m+1}$, where $m$ is a number of the free
variables of the virtual predicate obtained from this atom.
Consequently, each ground atom $p_i^k(t_1,...,t_k)$ is mapped into
$D_1$, and if there is any $t_i$, $1 \leq i \leq k$, such that
$I(t_i) = \circledR$, then $I(p_i^k(t_1,...,t_k)) = u_{\bot}$, where
$\bot \in X$ such that $\neg \bot = \bot$ is the
logic truth-value 'unknown'. In this way we guarantee that unmeaningful sentences as
$blu(Unicorn)$ will always have the 'unknown' logic truth-value. This
intensional interpretation can be given also to all contradictory
formulae (with truth-value $\top\in X$) that can not be nor true nor false, as the Liars paradoxes.

The main difference with standard FOL syntax is that here we can use
abstracted terms obtained from logic formulae, for example, "x believes that $\phi$" is given by formula $p_i^2(x,\lessdot \phi \gtrdot)$ ( where $p_i^2$ is binary "believe" predicate).

Let $\phi(\textbf{x})$ be any well-formed virtual predicate in $\L_{in}$,
then $\lessdot \phi(\textbf{x})\gtrdot_{\alpha}^{\beta}/g$  with the set of hidden variables $\alpha = (x_1,...,x_m)$, $m \geq 1$, is the  ground term (the assignment $g$ is applied only to free visible variables in $\beta = \overline{\textbf{x}}-\alpha$ whose semantics correlate is an intensional entity (concept) of degree $m$. If $m =0$, the intensional correlate of this singular ground term
is the proposition (sentence) "that $\phi$"; if $m =1$, the intensional
correlate is the property of "being something $x_1$ such that $\phi$";
if $m > 1$, then the intensional correlate is the concept
"the relation among $x_1,...,x_m$ such that $\phi$".

Certain complex nominative expressions (namely, gerundive and
infinitive phrases) are best represented as singular terms of the
sort provided by our generalized bracket notation $\lessdot
\phi(\textbf{x})\gtrdot_{x_1,...,x_m}^{\beta}$, where $m \geq 1$. This MV-intensional logic
$~\L_{in}$ differs from two-valued intensional FOL  in heaving these singular terms
$\lessdot \phi(\textbf{x})\gtrdot_{x_1,...,x_m}^{\beta}$ where the ground atoms $\phi(\textbf{x})/g$ are \emph{many-valued}.
\begin{definition} \label{def:abstrConv}
 An assignment $g:\V \rightarrow \D$ for variables in $\V$  in  $IFOL_B$ is
applied only to free variables in terms and formulae.
 Such an assignment $g \in \D^{\V}$ can be recursively uniquely extended into the assignment $g^*:\T \rightarrow \D$, where $\T$ denotes the set of all terms, by:
\begin{enumerate}
  \item $g^*(t_i) = g(x) \in \D$ if the term $t_i$ is a variable $x \in \V$.
  \item If the term $t_i$ is a constant $c \in F$ then $g^*(t_i)  \in \D$ us its tarskian interpretation.
  \item If a term $t_i$ is $f_i^k(t_1,...,t_k)$, where $f_i^k \in F$ is a
k-ary functional symbol and $t_1,...,t_k$ are terms, then
$g^*(f_i^k(t_1,...,t_k))$ is the value $u\in \D$ of this functions for the tuple of values in  $(g^*(t_1),...,g^*(t_k))$ or,
equivalently, in the graph-interpretation of the function,
  $(g^*(t_1),...,g^*(t_k),u) \in \D^{k+1}$.
   \item If $t$ is an abstracted term obtained for an open formula $\phi_i$, $\lessdot \phi_i(\textbf{x}_i) \gtrdot_{\alpha_i}^{\beta_i}$ where $\alpha_i \bigcup \beta_i$ are the variables in tuple $\textbf{x}_i$,then we must restrict the assignment to $g\in \D^{\beta_i}$ and to obtain recursive definition (when also $\phi_i(\textbf{x}_i)$ contains abstracted terms:
\begin{equation} \label{eq:assAbTerm}
  g^*(\lessdot \phi_i(\textbf{x}_i)\gtrdot_{\alpha_i}^{\beta_i}) =_{def}
    \left\{
    \begin{array}{ll}
   I(\phi_i(\textbf{x}_i))~~ \in D_{|\alpha_i|+1}, & \hbox{if  $\beta_i$ is  empty}\\
       I(\phi_i(\textbf{x}_i)[\beta_i
/g(\beta_i)])~~ \in D_{|\alpha_i|+1}, & \hbox{otherwise}
       \end{array}
  \right.
 \end{equation}
where $g(\beta) = g(\{y_1,..,y_m\}) = \{g(y_1),...,g(y_m)\}$ and $[\beta
/g(\beta)]$ is a uniform replacement of each i-th variable in the
set $\beta$ with the i-th constant in the set $g(\beta)$. Notice that $\alpha$ is the set of all free variables in the formula $\phi[\beta /g(\beta)]$.
\item  If $~t = \lessdot \phi_i\gtrdot$ is an abstracted term obtained from a sentence $\phi_i$ then \\$g^*(\lessdot \phi_i\gtrdot) = I(\phi_i) \in D_0$.
\end{enumerate}
\end{definition}
In what follows we will use the graph-interpretation for functions
 in FOL like its interpretation in intensional logics. We denote by $\L$ the set of all formulae $\phi$ of the logic $IFOL_B$, and denote by $~t_i/g~$ (or $\phi/g$) the ground term (or formula) without free variables, obtained by assignment $g$ from a term $t_i$ (or a formula $\phi$), and by  $\phi[x/t_i]$ the formula
obtained by  uniformly replacing $x$ by a term $t_i$ in open formula $\phi(x)$.\\
A \emph{sentence} is a (closed) formula having no free variables.
\begin{definition} \label{def:HerbrandValuat}
 The Herbrand base of a  logic $IFOL_B$ is defined by

 $~~H = \{p_i^k(t_1,..,t_k)~|~p_i^k \in P$ and $t_1,...,t_k$ are ground terms $\}$.
 \\
 Herbrand interpretations  are the mappings $v:H \rightarrow X$,which must satisfy the constraints for the  built-in ground atoms (any ground atom of a built-in predicates must have the same truth-value for every Herbrand interpretation) and for built-in propositional symbols $p_a$, for $a \in X$, $v(p_a) =a$.
  \end{definition}
  We are able to define the many-valued algebraic semantics of $\L_{in}$,
based on the standard extension of Herbrand interpretations to all
sentences $\L_0$ of $\L_{in}$.
\begin{definition}  \label{def:MV-algebra} \textsc{Belnap's 4-valued Semantics of $\L_{in}$}:\\
The algebraic semantics of the 4-valued many-sorted intensional first-order logic $\L_{in}$
(provided by Definition 11 in \cite{Majk25})  can be obtained by the unique extension of a given 4-valued Herbrand interpretation $v:H\rightarrow X$ (from Definition \ref{def:HerbrandValuat}), into the valuation  $v^*:\L_0 \rightarrow X$, where $\L_0 \subset \L$ is the strict subset of all sentences
(formulae without free variables), inductively as follows: \\for any formula $\phi, \psi \in \L$ and a given assignment $g:\V \rightarrow \D$,  we have that (here $\bigwedge$ and $\bigvee$ are the meet and joint operators of the lattice $X$, respectively)
\begin{enumerate}
  \item $~~v^*(\neg \phi/g) = ~\neg v^*(\phi/g)$, and for built-in propositions $v^*(p_a) = v(p_a)$ for $a\in X$,
  \item $~~v^*(\phi/g \odot \psi/g) =  v^*(\phi/g) \odot
v^*(\psi/g)$, \\for each logic connective $\odot \in \{\wedge, \vee,
\Rightarrow,\Leftrightarrow\}$, i.e., $\A_{B}$-operators in Definition \ref{def:mv-algebra} over Belnap's truth-lattice $X$,
  \item $~~v^*(((\exists_i) \phi)/g) =  v^*(\phi/g)$ if i-th variable $x$ is not a free  in $\phi$;\\
 $ = \bigvee \{v^*(\phi/g_1)~|~g_1 \in \D^{V}$ such that for all $y \in \V \backslash\{x\}, g_1(y) = g(y)\}$ otherwise,
  \item $~~v^*(((\forall_i) \phi)/g) =  v^*(\phi/g)$ if i-th variable $x$ is not a free  in $\phi$;\\
 $ = \bigwedge \{v^*(\phi/g_1)~|~g_1 \in \D^{V}$ such that for all $y \in \V \backslash\{x\}, g_1(y) =  g(y)\}$ otherwise.
\end{enumerate}
 We denote by $\I_{MV} \subseteq X^{\L_0}$ the set of all many-valued valuations
 of $\L_{in}$ that have fixed (invariant) interpretation for each built-in predicate.\footnote{As, for example, the identity predicate $p_1^2 \in P$
 or nullary predicate letters $\{p_a~|~ a \in X\}$ for which it must be satisfied that for any
 $v^* \in \I_{MV}, v^*(p_a) = a$,  $v^*(p_1^2 (u_1,u_2)) = t$ iff $(u_1,u_2,t) \in  R_{=}$, and $v^*(p_2^2 (u_1,u_2)) = t$ iff $(u_1,u_2,t) \in
 R_{=_{in}}$,  and for any assignment $g$,
 $v^*(p_1^1 (t_i)) = v^*(T(t_i)) =_{def}  v^*(\phi/g)$ if term $t_i = \lessdot \phi/g\gtrdot$; $f$ otherwise.\\
 The difference between  $\I_{MV}$ and the total set of Herbrand interpretations $X^H$ is caused by the presence of the built-in predicates.}
\end{definition}
We recall that the set-based (for infinite sets as well) of the operators $\bigwedge$ and $\bigvee$ is well defined because our many-valued logics are based on the complete (and distributive) lattices $(X, \leq)$, which satisfy these requirements.

Notice that this Belnap's 4-valued extended interpretation $v^*:\L_0 \rightarrow X$ defined above \emph{is not a homomorphism}, because of the point 3 (and 4 as well), where the truth of the closed formula $(\forall x_i) \phi(x_i)$ can not be obtained from the logic value of $\phi(x_i)$ from the fact that  to any formula with free variables we can not associate any logic value. Because of that, we define a new version of many-valued interpretation, denominated "MV-interpretation" (which is also a homomorphism \cite{Majk25} between syntax algebra of $IFOL_B$ logic and algebra of relations in $\mathfrak{Rm}$):
\begin{definition} \label{def:MV-interpret} \textsc{MV-interpretations}:
We define, for a valuation $v^*:\L_0\rightarrow X$ of the sentences in $\L_0 \subset \L$ of the 4-valued many-sorted intensional first-order logic  $\L_{in}$,  the MV-interpretation $I^*_{B}:\L_0 \rightarrow \mathfrak{Rm}$, such that for any sentence $\phi/g\in \L_0$,
 \begin{equation}\label{eq:mvR}
 I^*_{B}(\phi/g) =
 \left\{
    \begin{array}{ll}
 \{v^*(\phi/g)\} \in \widetilde{X} \subset \mathfrak{Rm}, & \hbox{if $~~v^*(\phi/g) \neq \bot$}\\
      \emptyset, & \hbox{otherwise}
    \end{array}
  \right.
 \end{equation}
 and we define also the unique extension of $I^*_{B}$ to all open formulae $\L$ in
$\L_{in}$ as well, such that for any open formula (virtual predicate) $\phi(x_1,...,x_k) \in \L$,
\begin{equation} \label{eq:openEQ}
I^*_{B}(\phi(x_1,...,x_k)) = \{(g(x_1),...,g(x_k),a)~|~g \in
\D^{\V}~~ and~~ a = v^*(\phi/g)  \neq \bot\}\in \mathfrak{Rm}
\end{equation}8
We denote by $\W$ the set of all MV-interpretations derived from the set of 4-valued interpretations $v^*\in \I_{MV}$ specified in Definition \ref{def:MV-algebra}, with bijection $is_{MV}:\I_{MV} \simeq \W$ such that for any $v^* \in \I_{MV}$ we have that $I^*_{B} = is_{MV}(v^*):\L\rightarrow \mathfrak{Rm}$.
\end{definition}
For each \emph{fixed intensional} interpretation $I:\L \rightarrow \D$ and valuation  $v^*:\L_0 \rightarrow X$ in Definition  \ref{def:MV-algebra} that respect all built-in predicates, the corresponding extensionalization function $h:\D \rightarrow \mathfrak{Rm}$,  for any virtual predicate $\phi(x_1,...,x_n) \in \L$, satisfies the following 4-valued generalization of Tarski's FOL constraint:
\begin{multline} \label{eq:MVtarski}
 h(I(\phi(x_1,...,x_n))) =_{def} \\\{(g(x_1),...,g(x_n),a)~|~g \in
 \D^{\V},  h(I(\phi(x_1,...,x_n)/g)) = \{a\}, a\neq \bot \}
\end{multline}
  with $h(I(\phi(x_1,...,x_n)/g)) = I^*_{B}(\phi(x_1,...,x_n)/g)$,
so that  holds \cite{Majk25} the general composition $I^*_{B} = h\circ I$.
\section{Closed Knowledge Assumption (CKA) and Learning}
First of all, the CKA \emph{is not a well-known Closed World Assumption} (CWA) which postulates that all sentences of a logic that are not proven to be  true (the lack of knowledge) \emph{must be false}.  I argue that this conflation is often inappropriate, especially for: distributed knowledge bases,
semantic web applications, probabilistic and intensional reasoning,
systems with incomplete information. Instead, CKA can be summarized as:

"\emph{Only the explicitly represented knowledge is considered known; what is not represented is simply not known, rather than automatically false.}"

 The emphasis shifts from truth to knowledge. My work on intensional logic treats propositions as objects of reasoning. In that setting, it is natural to distinguish: whether a proposition is true, whether it is believed,
whether it is known, whether its probability is defined.
This distinction is much closer to epistemic logic than to classical logic programming.

The main strengths of CKA are: It separates knowledge from truth.
It avoids some unintuitive consequences of CWA in incomplete domains.
It aligns well with intensional and epistemic logics.
It fits naturally with probabilistic and temporal reasoning (which will be introduced in $IFOL_B$ in next dedicated paper), where incomplete information is common.
From a theoretical perspective, CKA is a coherent attempt to provide a cleaner epistemic foundation for reasoning under incomplete information. It is especially consistent with my broader program of treating propositions, probabilities, and temporal facts as intensional objects rather than reducing them to simple truth values. While it has had limited impact on mainstream AI systems, it addresses a genuine conceptual issue that arises when combining logic, probability, and knowledge representation.

 CWA principle has been valid only for the classical 2-valued logic, where the sentences can be true or false only.  In CKA we are working with the 4-valued logic with the logic value "unknown" (denoted by the symbol $\bot$) as well, so that only for such logics we can define CKA, that is, that all sentences for which we do not know their logic values, that they are unknown sentences (that is assigning to them the logic value $\bot$).

Knowledge representation formalisms are aimed to represent general conceptual information and are typically used in the construction of the knowledge base of reasoning agent (a robot in our case). A knowledge base can be thought of as representing the beliefs of such an agent.

Differently from 2-valued logic where extension of a given predicate is composed by only \emph{true} atoms, here for MV-logic $IFOL_B$ the ground atoms of extension of this predicate can have three truth-values (\emph{except the unknown value} $\bot$ for unknown facts which are not explicitly represented). In this way, the complete knowledge of AGI robots would correspond to the extensions of the only known n-ary \emph{concept}, $n \geq 1$, and hence it would be very easy to compare the knowledge of different robots developed by their individual experiences supported by their deductive inference capabilities.\\
\textbf{Remark:} The logic knowledge (composed by the logic sentences and corresponding intensional L-concepts heaving as extensions only the logic truth-values) will be provided in next Section both with robot's autoepistemic deduction.
\\$\square$\\
Consequently, we are able to represent the whole \emph{AGI robot's knowledge Database} of its n-ary concepts, $n \geq 1$, in each fixed instance of time (it is an analog to traditional relational Database with standard 2-valued FOL) as follows:
\begin{definition}\label{Def:knowledge} \textsc{AGI Robot's Current Atomic Knowledge Database}\footnote{This definition is different from the definition in \cite{Majk25}, in order to be able also to derive the current Herbrand model of robot's knowledge $v:H \rightarrow X$, where $H$ is the Herbrand base for all  predicates in $P$, Note that the autoepistemic predicate $Know$ is a meta-èredicate, so that $Know \notin P$.}:\\
For a given instance of time, the current AGI Robot knowledge is defined by its current MV-interpretation $\textbf{I}^*_{B}$ (that represents robot's current world $w \in \W$) defines  the extension of robot's knowledge Database $\K$ by:
\begin{equation} \label{eq:two-know}
~~ \K ~= ~ \{R_{p_i^k} = \textbf{I}^*_{B}(p_i^k(x_1,...,x_k) ~|~p_i^k \in P, ~for~ x_i \in \V,  1\leq i\leq k\}
\end{equation}
where $\textbf{I}^*_{B}$ is the current  MV-interpretation (a function from $\L$ to $\mathfrak{Rm}$) and $\V$ the set of variables and $R_{p_i^k}$ is the $(k+1)$-ary relation obtained for the k-ary predicate $p_i^k$.
\end{definition}
The Belnap's bilattice has both truth and  the knowledge orderings.  The truth ordering we used for the set of relations in $\mathfrak{Rm}$  in Definition \ref{def:POrder}, while the knowledge ordering is useful just for the knowledge base and the principle of CKA: Each k-ary tuple, $k\geq 2$, not in knowledge Database $\K$ is considered as unknown fact.
Thus, the inclusion relation for robot's knowledge is just the simple set-inclusion $\subseteq$.\\
 Consequently, the atomic knowledge database is the subset of current metaknowledge of a robot is the set of relations, that is,
\begin{equation} \label{eq:two-metaknow}
\K \subset~~  ~ \{R ~|~R  \in Im(\textbf{I}^*_{B}), ~for~ ar(R) \geq 2\}
\end{equation}
where for each tuple of ground terms, $\textbf{d} = (t_1,...,t_k,a) \in R_{p_i^k} \in \K$, of relation with arity $k+1\geq 2$, the current truth-value of this, for robot known, fact is equal to last value of this tuple, $a =\pi_{k+1}(\textbf{d}) \in \{f, \top, t\}$ with truth-ordering $f < \top < t$. So, for any ground instance of each robot's (real or virtual) predicate (and corresponding intensional concept), robot is able to know the level of truth in a given instance of time. Thus, given current atomic knowledge database $\K$, we are able to derive from it the current Herbrand model $v:H \rightarrow X$ and its extension $v^*$ to all sentences, for robot's knowledge by, for any ground atom $p_i^k(t_1,...,t_k)\in H$, with the ground terms $t_i$, for $1\leq i\leq k$,
\begin{equation} \label{eq:two-metaknowH}
 v(p_i^k(t_1,...,t_k)) =
 \left\{
    \begin{array}{ll}
 a\in \{f, \top, t\}, & \hbox{if $~~(t_1,...,t_k,a) \in R_{p_i^k} \in \K$}\\
      \bot, & \hbox{otherwise}
    \end{array}
  \right.
\end{equation}
Thus, from current atomic knowledge database $\K$, we are able to derive current Herbrand model $v$, its extension to all sentences $v^*$ and thence the current possible world (MV-model in Definition \ref{Def:knowledge}) $\textbf{I}^*_B = h\circ I \in \W$.

By default the initial robot's atomic knowledge database $\K$ before the initial process of learning contains only the extensions of the built-in predicates (invariant w.r.t any Herbrand interpretation $v$. However, we intend to set this initial knowledge by some necessary quantity of knowledge which will provide enough information from which a robot will be able to obtained its own knowledge by intensive initial learning and own experience.

In the accelerating world of AI, discussions often gravitate toward its potential to surpass human intelligence and the hypothetical scenario of machines overtaking decision-making processes, possibly leading to catastrophic consequences. This fear is not entirely misplaced. Today, AI technologies such as self-driving cars and smart homes already make crucial data-based decisions without direct human intervention. However, how AI systems interpret and process the data they have often mitigates the risks of these technologies. One key element in this process is Closed World Assumption in AI.

In \emph{2-valued logic}, the Closed World Assumption (CWA) asserts that if a fact isn't recorded as true, then it is considered false. So, the knowledge base or system has complete information about the domain, and anything not explicitly stated is assumed to be non-existent or incorrect. In particular, we drew attention to the problem of representing negative information — information about the non-existence of objects or events or processes or influences. Negative information, we argued, must sometimes be encoded by the mere absence of positive information to the contrary: propositions which are not given, or which are not deducible from those which are given, should be assumed to be false. The closed world assumption is used in many practical applications of logic (planning, prolog\footnote{In Prolog, \textbf{Yes} means a statement is provably true. Consequently,
No means a statement is not provably true. This only means that
such a statement is false, if we assume that all relevant information
is present in the respective Prolog program.
For the semantics of Prolog programs we usually do make this
assumption. It is called the Closed World Assumption: we assume
that nothing outside the world described by a particular Prolog
program exists (is true). }, databases\footnote{The Closed-World Assumption (CWA) on a database expresses that an atom not in the database is false. The CWA is only applicable in domains where the database has complete knowledge. In many cases, for example in the context of distributed databases, a data source has only complete knowledge about part of the domain of discourse.}) in order to lower the complexity class of the reasoning algorithms used. This assumes that any predicates that are stated are true and the ones that aren’t are false.

Based on these facts we are able, in analogy to the Closed World Assumption (CWA) used for Databases defined by the standard 2-valued FOL (as, for example the RDBs), we can introduced the Closed Knowledge Assumption (CKA) and to show that it is satisfied for the  4-valued many-sorted Intensional FOL.
\begin{definition}\label{Def:CKA} \textsc{Closed Knowledge Assumption(CKA)}:\\
The CKA for the  many-sorted Intensional FOL based on Belnap's 4-valued bilattice of truth-values is defined, for every Herbrand interpretation $v:H \rightarrow X$ and assignment $g:\V \rightarrow \D$, as follows:\\
For each k-ary virtual predicate $\phi(x_1,...,x_k)$, $k \geq 1$,

$v^*(\phi(x_1,...,x_k)/g) \neq \bot~~$ iff $~~(g(x_1),...,g(x_k))\in \pi_{-i}(I_B^*(\phi(x_1,...,x_k)))$.
\end{definition}
So, we obtain the following property for the robot's knowledge database $\K$ based on the logic $IFOL_B$:
\begin{coro} The robot's atomic knowledge database $\K$ based on the logic Belnap's 4-valued many-sorted intensional first-order logic $IFOL_B$, representing the Herbrand many-valued interpretation $v:H \rightarrow X$ at any given instance of time, satisfies the Closed Knowledge Assumption.
\end{coro}
\textbf{Proof}:
In this case the general virtual predicates are just the predicates $p_i^k \in P$ from which is defined the Herbrand base and, from $\K$ and  (\ref{eq:two-metaknowH}) its Herbrand interpretation $v:H \rightarrow X$.
So, for each sentence of  predicate $p_i^k\in P$,  for a given assignment $g\in \D^\V$,
\\if $a= v(p_i^k(x_1,...,x_k)/g) \neq \bot~~$ then $~~(g(x_1),...,g(x_k),a)\in R_{p_i^k} \in \K$, \\
that is, $(g(x_1),...,g(x_k)) \in \pi_{-i}R_{p_i^k} =\pi_{-i}(I_B^*(p_i^k(x_1,...,x_k)))$ \\
as required from CKA.
\\$\square$\\
With this CKA property of its knowledge database $\K$, a robot is able to verify which  of its sentences are still unknown, that is, to derive the logic truth value of every sentence of its logic language $\L$.

This result is analogous to the property of Closed World Assumption for RDB (Relational Databases) based on 2-valued standard FOL (First-order Logic).

By using in our knowledge 4-levels cognitive neurosymbolic structure also a neural LLM (Large Language Models), we are able to obtain the parsing from the natural language expressions into the FOL formulas, and to consider LLM for natural interactive dialog between a robot and humans (as in ChatGPT and similar applications).  The exploration of the capability of neural models in parsing English sentences to FOL are provide, for example, in \cite{SAKr20}. In next section we will provide the relationships between logic predicates and n-ary concepts which names are just the simple words or finite natural language expressions, so that we can use also the parsing of natural language expressions into into intensional expressions of these concepts as well.\\\\
\textbf{Conceptual learning bridge}: We consider that LLM with common statistical natural language knowledge and initial sets of predicates and corresponding concepts are part of initial robot's knowledge.  However, robot will be able to define new predicates (when with given set of predicates in $P$ and their corresponding intensional entities (concepts) is not possible to generate a logic formula for particular natural language expressions), or modify existing predicates in $P$  (by adding new attributes to such predicates).  So, with this processes we obtain the scalability of the symbolic AI based on IFOL predicates and their domains.

The component of the neural AI learning and its scalability is given by the properties of implementation of LLM.\\\\
\textbf{Learning from neural to symbolic levels}:  Each personal robot's experience can be internally represented by some logic sentence.  \\
For example,
 \textbf{The Labeling Process by Reification}:\\
To label a motor program (e.g., a routine that moves the right arm to pick up a block), the robot uses an Intensional Abstraction Operator.\\
\begin{enumerate}
  \item \textbf{The Program}: Let’s say the raw code for moving an arm is $P_{1}$ and motor programs be defined by true atoms of the binary predicate $MotPrg(x_1,x_2)$ where the variable $x_1 \in \V$ defines the row codes in neural system and $x_2 \in \V$ defines the labels (names) of these row codes, so that for this assignment $g:\V \rightarrow \D$ forb variables in $\V$, $g(x_1) = P_{1}$ and  $g(x_2) = "moving ~right~ arm"$ for which the ground atom $MotPrg(x_1,x_2)/g = MotPrg(P_1,moving ~right~ arm)$ is true.\\
  \item \textbf{The Logic Term}: The robot uses the abstraction operator $\lessdot\_ \gtrdot$ to create a "name" for this code. The label becomes an intensional entity - a symbol the robot can think about without actually running the code. In fact by transformation of this logic predicate into abstracted term, we obtain (by using the intensional mapping $I$) that

      $u= g^*(\lessdot MotPrg(x_1,x_2)\gtrdot^{x_1}_{x_2}) = I(MotPrg(g(x_1),x_2))$

  $= I(MotPrg(P_1,x_2)) \in D_1$ is an intensional entity (an unary concept) such that its extension (for extensionalization function $h$) is just a singleton, i.e., $h(u) = \{moving ~right~ arm\}$. Thus, the intensional entity $u = I(MotPrg(P_1,x_2))$ can be used as a label (name) for the raw code $P_1$.\\
  \item \textbf{The Predicate}: The robot then uses a ternary predicate like $Exec(y_1,y_2, \lessdot \phi(\textbf{x})\gtrdot^\beta_\alpha)$ with $g(y_1) = "in~ present"$ and $g(y_2) = \textbf{I}$, such that

       $Exec(y_1,y_2, \lessdot MotPrg(x_1,x_2)\gtrdot^{x_1}_{x_2})/g = Exec(in~ present, \textbf{I}, I(MotPrg(P_1,x_2)))$ \\
        translates to:

  "\emph{\textbf{I} (the coordinator) am currently executing the motor program of raw code} $P_1$."\\

\item \textbf{Reaching Conscious Knowledge}: this sentence in a temporary memory of symbolic AI level of robot can be inserted into conscious robot's knowledge by epistemic predicate $Know$ and can be used for other logic deductions.
\end{enumerate}
 By this personal, based on robots experiences, learning process, the \emph{scalability} of the robot's conscious knowledge is guaranteed by this bridge from the neural to symbolic robot's knowledge.
\section{The 4-valued Autoepistemic Deduction of AGI Robots}
Differently from the the knowledge Database of robot's intensional n-ary concepts (provided in previous Section), the robots current \emph{logic knowledge} (composed by ground atoms of the distinguished predicate $Know$ in Definition \ref{def:syntax}) is directly derived from its experiences (based on its neuro-system processes that robot is using in this actual world), in an analog way as human brain does:
  \begin{itemize}
    \item As an activation (under robot's attention) of its neuro-system process, as a consequence of some human command to execute some particular job.
    \item As an activation of some process under current attention of robot, which is part of some complex plan of robot's activities connected with its general objectives and services.
  \end{itemize}
  In both cases, for a given assignment of variables $g:\V\rightarrow \D$ of virtual predicate $\psi(\textbf{x})$ with the \emph{set} of variables  $\beta =\overline{\textbf{x}} $, which  \emph{L-concept} $I(\psi(\textbf{x})/g) \in D_{1}$ is grounded by this particular process, is transformed into abstracted term and hence robot's knowledge system generates the new ground knowledge atom with three terms $t_i$, $i = 1,2,3$, $Know(t_1,t_2,t_3)/g$ with  abstracted term $t_3 = \lessdot \psi(\textbf{x})\gtrdot^\beta$   in robot's temporary memory.
  So, the robot's logic knowledge, which partecipates in deduction processes as well, is composed by a set of known robot's sentences  $\psi(\textbf{x})/g$ heaving logic truth-values of Belnap's billattice.\\
  Consequently,  the  explicit (conscious) robot's  logic knowledge in actual world  (current time-instance) here is represented by the ground atoms of the distinguished  $Know$ predicate in Definition \ref{def:syntax},
 \begin{equation} \label{eq:esem2}
  Know(t_1,t_2,\lessdot \psi(\textbf{x})\gtrdot^\beta)/g = Know(g^*(t_1),g^*(t_2), g^*(\lessdot \psi(\textbf{x})\gtrdot^\beta))
  \end{equation}
  such that $g^*(t_1) = in ~ present$ and $g^*(t_2) = me$ (the robot itself),  for the extended assignments $g^*:\T \rightarrow \D$ in Definition \ref{def:abstrConv} to all terms in $\T$.

  We recall that each robot's extensionalization function $h$ in (\ref{eq:dueM6}) is indexed by the time-instance. The actual robot's world extensionalization function (in the current instance of time) is denoted by $\hbar$, and determines the current robot's knowledge.
 Clearly, the robots knowledge changes in time and hence determines the extensionalization function $h $ in any given instance of time, based on robots experiences. Thus, as for humans, also the robot's knowledge and logic is a kind of temporal logic, and evolves with time.
 \\
   \textbf{Remark}:  We consider that only robot's experiences (under robot's attention) are transformed into the ground atoms of the $Know$ predicate, and the required (by robot) deductions from them (by using $IFOL_B$ deduction extended by the three epistemic axioms in what follows) are transformed into ground atoms of $Know$ predicate, and hence are saved in robot's temporary memory as a part of robot's \emph{conscience}.\\ Some background process (unconscious for the robot) would successively transform these temporary memory knowledge into permanent robot's knowledge in an analog way as it happen for humans.
   \\$\square$\\
 Based on these considerations, we can specify exactly what are semantically the robot's knowledge models:
 \begin{definition} \label{def:Kmodels} \textsc{Logic Knowledge Models}:\\
 Herbrand interpretations  $v:H \rightarrow X$ in Definition \ref{def:HerbrandValuat} are denominated \textsl{knowledge interpretations}, if for each ground atom in (\ref{eq:esem2}) of the $T$  and $Know$ predicates the logic truth-value of them is equal just to the truth-value of the  sentence reified in the ground atoms of these two predicates as an abstracted term,
 \begin{multline} \label{eq:esem3s}
 v^*(Know(t_1,t_2,\lessdot \psi(\textbf{x})\gtrdot^\beta)/g ) = v^*(\psi(\textbf{x})/g)\in X\\
  v^*(T(\lessdot \psi(\textbf{x})/g \gtrdot)) = v^*(\psi(\textbf{x})/g) \in X~~~~~~~~~~~~~~~~~~~~~~~~~~~~~~~~~~~~~~~~~~~~~~~~~~~~~~~~
 \end{multline}
 \end{definition}
Thus, in the actual world $h$ (actual time-instance), the known fact (\ref{eq:esem2}) for robot becomes the ground atom
 \begin{equation} \label{eq:esem3}
  Know(t_1,t_2,\lessdot \psi(\textbf{x})\gtrdot^\beta)/g = Know(in ~ present,me, I(\psi(\textbf{x})/g))
  \end{equation}
  with a L-concept $I(\psi(\textbf{x})/g)) \in D_1$  and, from the fact that $\textbf{I}_B = h\circ I$, we obtain
 \begin{equation} \label{eq:esem3sb}
  h(I(\psi(\textbf{x})/g)))=\textbf{I}_B(\psi(\textbf{x})/g ) =
 \left\{
    \begin{array}{ll}
 \{a\} \in \widetilde{X} \subset \mathfrak{Rm}, & \hbox{if $~~a = v^*(\psi(\textbf{x})/g) \neq \bot$}\\
      \emptyset, & \hbox{otherwise}
    \end{array}
  \right.
  \end{equation}
  While, for the L-concept $u=I(Know(t_1,t_2,\lessdot \psi(\textbf{x})\gtrdot^\beta)/g )\in D_1$, we obtain \\
   $\hbar(I(Know(t_1,t_2,\lessdot \psi(\textbf{x})\gtrdot^\beta)/g )) =
 \left\{
    \begin{array}{ll}
 \{v^*(\psi(\textbf{x})/g)\} \in \widetilde{X}, & \hbox{if $~~v^*(\psi(\textbf{x})/g) \neq \bot$}\\
      \emptyset, & \hbox{otherwise}
    \end{array}
  \right.$\\
 that is, for this \emph{current} MV-model $\textbf{I}_B=h\circ I$, and a ground atom of the $Know$ predicate, we obtain
 \begin{equation} \label{eq:esem3s1}
 \textbf{I}_B(Know(t_1,t_2,\lessdot \psi(\textbf{x})\gtrdot^\beta)/g ) =
 \left\{
    \begin{array}{ll}
 \{a\} \in \widetilde{X} \subset \mathfrak{Rm}, & \hbox{if $~~a = v^*(\psi(\textbf{x})/g) \neq \bot$}\\
      \emptyset, & \hbox{otherwise}
    \end{array}
  \right.
 \end{equation}
 and consequently from (\ref{eq:esem3sb}) and (\ref{eq:esem3s1}), in accordance with (\ref{eq:esem3s}), we obtain the  MV-model knowledge-constraint
 \begin{equation} \label{eq:esem3MVc}
 \textbf{I}_B(Know(t_1,t_2,\lessdot \psi(\textbf{x})\gtrdot^\beta)/g ) = \textbf{I}_B(\psi(\textbf{x})/g)
 \end{equation}
Thus, we defined by using the predicate $Know$ what the robot knows about each its sentence $\psi(\textbf{x})/g$. However, we have to define also the opposite feature about natural language sentences that use the "do not know" in all temporal conditions (does not know, did not know, will not know), generally expressed in natural language by:

"I do not know if $\psi(\textbf{x})/g$",

"I do not know weather $\psi(\textbf{x})/g$"\\
just because their translation into a logic formula of $IFOL_B$ needs to use this $Know$ predicate as well, and the fact that this $\phi(\textbf{x})/g$ has to be \emph{unknown}, that is $v^*(\phi(\textbf{x})/g) = \bot$. It is clear that such logic formula must be different from $\neg Know(t_1,t_2,\lessdot \psi(\textbf{x})\gtrdot^\beta)/g$ because $\neg \bot \neq t$.

In accordance with the Closed Knowledge Assumption, if we do not know something, this something must be \emph{unknown}.
Consequently, this formula is provided by (see the 2-valued logic truth-scheme formulae for any truth-value $a \in X$, given by (3) in \cite{Majk25})
\begin{equation} \label{dnKnow}
p_\bot \Leftrightarrow Know(t_1,t_2,\lessdot \psi(\textbf{x})\gtrdot^\beta)/g
\end{equation}
where $p_\bot$ is the built-in propositional letter such that for each many-valued valuation $v^*(p_\bot) = \bot$. In fact, the sentence in (\ref{dnKnow}) is true only if $v^*(Know(t_1,t_2,\lessdot \psi(\textbf{x})\gtrdot^\beta)/g) = \bot$, that is, from (\ref{eq:esem3s}),  iff $v^*(\psi(\textbf{x})/g) = \bot$ is unknown.\\
So the formula (\ref{dnKnow}) can be expressed  in English, for example, by:

" I \emph{do not know} wether $\psi(\textbf{x})/g$", \\
if $g^*(t_1) = me$ and $g^*(t_2) = in ~present$.   Otherwise by

" $g^*(t_1)$ \emph{does not know} wether $\psi(\textbf{x})/g$",  ~~~~~~if $g^*(t_2) = in ~present$

" $g^*(t_1)$ \emph{did not know} wether $\psi(\textbf{x})/g$",  ~~~~~~~~if $g^*(t_2) = in ~past$

" $g^*(t_1)$ \emph{will not know} wether $\psi(\textbf{x})/g$",  ~~~~~~~~if $g^*(t_2) = in ~future$.
\\\\
%
\textbf{Remark}: Note that for the assignments $g:\V\rightarrow \D$, such that $g^*(t_1)= in ~future$ and $g^*(t_2)= me$ we consider robot's hypothetical knowledge in future, while in the cases when $g^*(t_1) = in ~past$ we consider what was robot's knowledge in the past. Consequently, generally the predicates of IFOL for robots, based on the dynamic changes of its knowledge has to be indexed by the time-instances (which are possible worlds of IFOL), for example by using an additional predicate's argument for them. \\
 In the examples in next, we will consider only the case of robot's current knowledge (in the actual world with extensional function $h$) when $g^*(t_1) =in ~present$, so we avoid the introduction of the time-instance variable for the predicates; only at the remark at the end of this section we will show how to use time-variable $\tau$.
\\$\square$\\
In what follows, the robot's knowledge system will use only the knowledge interpretations $v:H\rightarrow X$ and their extensions $v^*:\L_0\rightarrow X$ as described by Definition \ref{def:Kmodels}, and we denote by $\mathbb{V}$ the set of all such  knowledge interpretations.
 As an inference many-valued system, we can use the truth-preserving entailment provided by Definition 57 in Chapter 5 of \cite{Majk22}, here rewritten using notations used in this paper by:
\begin{definition}\cite{Majk22} \label{def:seqSatisf} \textsc{New truth-preserving entailment:} Let $\Gamma = \{\phi_1,..,\phi_n \} \subset \L_0$ be a set of sentences and we define the truth-preserving entailment in two alternative cases:

1. \emph{Sequent-based entailment} with sequents
$\Gamma \vdash \psi$. We say that  $v \in \mathbb{V}$ \emph{satisfies} this sequent iff $v^*(\phi_1) \bigwedge...\bigwedge v^*(\phi_n)\leq v^*(\psi)$,  where $\bigwedge$ is the meet operator in lattice $(X.\leq)$, and we denote this sequent satisfaction by  $\Gamma \models_v \psi$. A sequent is \textsl{valid}  if it is satisfied for all  $v \in \mathbb{V}$.

We denote by $\phi \preceq \psi$ iff for all $v \in \mathbb{V}$, $v^*(\phi) \leq v^*(\psi)$.

Valuation $v$ satisfies a rule  $\frac{s_1, ..., s_k}{s}$ iff $v$ satisfies the conclusion sequent $s$ of this rule whenever it satisfies all sequent premises
$s_1, ..., s_k$ of this rule. Then, a \emph{model} of this logic is any valuation $v\in \mathbb{V}$ which satisfies all logic sequent rules of this logic (the structural sequent rules as  Cut, Weakening, etc.. are satisfied by all valuations).

2. \emph{Model-based entailment} where $\Gamma$ is the set of theses $\phi_i\in \L_0$ with associated prefixed low-boundary truth values $a_i \in X$,  such that their \emph{models} are defined as a subset of all knowledge interpretations in $\mathbb{V}$:
\begin{equation} \label{eq:MVmodels}
\mathbb{V}_{\Gamma} = _{def} \{v \in \mathbb{V} \mid \forall \phi_i \in \Gamma. (v^*(\phi_i) \geq a_i)\}~~ \subset \mathbb{V}
\end{equation}
that is, a knowledge interpretation $v \in \mathbb{V}$ is a model of the theses in $\Gamma$ if the truth-value of each sentence $\phi_i$ in $\Gamma$ is greater or equal to its specified low-boundary truth-value $a_i$.
So, we introduce the model-based truth-preserving entailments of a sentence $\psi$ from the theses in $\Gamma$  and $v^*\in \mathbb{V}_{\Gamma}$ by
\begin{multline} \label{eq:MVtruth-preserv}
\Gamma \models_v \psi~~~ \emph{iff } ~~~(\exists \phi_i \in \Gamma). (  v^*(\psi)\geq v^*(\phi_i))\\
\Gamma \models \psi~~~ \emph{iff } ~~~ (\forall v\in \mathbb{V}_{\Gamma}).\Gamma \models_v \psi~~~~~~~~~~~~~~~~~~
\end{multline}
\end{definition}
Note that (\ref{eq:MVmodels}) is a generalization of the standard 2-valued logic case when all $a_i = t$ are  true.
\\
\textbf{Remark}: In this moment we are not interested to the possibility of  a robot to derive by himself the paradoxes, like Liar paradox discussed in next section that requires the high level mathematical ability (by using G$\ddot{o}$del numbering and creation by himself of new predicates like $Diag$ and Peano arithmetic FOL theorems, for example. In this stage it is enough only to support the existence of paradoxes by robot's knowledge base without any logic problem (such paradoxes can be provided by humans in initial robot's knowledge like all ground atoms of built-in predicates).
\\$\square$\\
 So,  for the sentences $\phi$  that are paradoxes (Like the Liar paradox discussed in next section) the inconsistent truth-value $\top$ is used:
 \begin{definition} \label{def:paradox}
 A \textsl{paradox} can be considered as particular built-in  sentence $\phi \in \L$ (logic formula which is not an atom, without free variables) that for each knowledge interpretation $v \in \mathbb{V}$ it is always an inconsistent sentence, that is, $v^*(\phi) = \top$.
\end{definition}
Inconsistent sentences are not only the paradoxes. Consider for example a ground atom (of non built-in or distinguished predicates $T$ and $Know$) that previously robot learned to be true and successively learned to be false (and viceversa), so to them will be assigned by knowledge interpretations the inconsistent value $$\top = f \vee t$$ that will remain such for all future time of robot's life.

 So, robots can support in their knowledge also the built-in paradoxes and other inconsistent ground atoms or sentences, and to compose them with other sentences by using logic connectives and to compute by $v^*$ the truth-values of such composed formulae as well.
  For example, let $\phi$ be an inconsistent sentence and $\psi$ not: if for given knowledge interpretation $v\in \mathbb{V}$, $v^*(\psi) = f$ then $v^*(\phi \wedge \psi) = v^*(\phi) \wedge v^*(\psi)=\top \wedge f =f$ (is not more inconsistent) while if for another $v_1\in \mathbb{V}$, $v^*_1(\psi) = t$ then $v^*_1(\phi \wedge \psi) = v^*_1(\phi) \wedge v^*_1(\psi)= \top \wedge t = \top $ is inconsistent.

Notice that the entailment $\Gamma \models \psi$ is \emph{valid} entailment differently from the specific entailment $\Gamma \models_v \psi$ for a given $v \in \mathbb{V}_{\Gamma}$, used also by Ginsberg's truth preserving entailment in point 3 of Definition \ref{def:billat2} when $\Gamma$ is a singleton set, as we can show:
\begin{coro}
Let $\Gamma$ be a singleton composed by only $\phi \in \L_0$. The valid entailment $\Gamma \models \psi$ can be replaced by $\phi \models \psi$ so that (\ref{eq:MVtruth-preserv}) reduces to:
$$\phi \models \psi~~~ \emph{iff } ~~~
 (\forall v\in \mathbb{V}_{\Gamma}).\phi \models_v \psi$$
and each knowledge interpretation  $v\in \mathbb{V}$ (and hence each model of $\phi$ as well) is closed by Ginsberg's specification in Definition \ref{def:billat2}.
\end{coro}
\textbf{Proof}: It is easy to verify that all three points in Definition \ref{def:billat2} are satisfied for each knowledge interpretation  $v\in \mathbb{V}$.
\\$\square$\\
Let us consider now in which way the many-valued entailment $\Gamma \models_v \psi$ generalizes the standard 2-valued FOL entailment (where all sentences in $\Gamma$ and conclusion $\psi$ must be true) for  the FOL inference rules for logic quantifiers and Modus Ponens $\phi, \phi \Rightarrow \psi \models_v \psi$. By this generalization in this many-valued intensional FOL we mean that not all premises in $\Gamma$ and conclusion $\psi$ must be true for this given interpretation $v^*$. So, we will only relax the Modus Ponens rule of the standard 2-valued FOL into this MV-Modus Pones rule in which only the implication $\phi \Rightarrow \psi$ must be true and not the thesis $\phi$ and the conclusion $\psi$:
\begin{definition} \label{def:MP}
For a given knowledge interpretation $v^*$ the MV-Modus Ponens rule (entailment) is defined as follows:
\begin{equation} \label{def:MP}
if~~ v^*(\phi \Rightarrow \psi) = t ~~~~then~~~~ \phi, \phi \Rightarrow \psi \models_v \psi
\end{equation}
The inference rules for existential and universal quantifiers used in standard 2-valued FOL as for example Universal Instantiation  entailment $(\forall x)\phi(x) \models_v \phi(x)/g$, are valid also in this many-valued many-sorted intensional $FOL_B$.
\end{definition}
This definition for inference rules for quantifiers are obviously valid in $FOL_B$, where, for example the Universal Instantiation entailment is satisfied from the fact that from point 4 in Definition \ref{def:MV-algebra} it holds that $v^*(\forall x)\phi(x) ) \leq v^*(\phi(x)/g)$ also if we do not impose the standard 2-valued FOL requirement that $v^*(\forall x)\phi(x) ) = t$. We only have to show that the MV-Modus Ponens rule in (\ref{def:MP}) is well defined generalization of standard Modus Ponens rule in 2-valued FOL.
\begin{coro} \label{coro:MP}
The MV-Modus Ponens rule in (\ref{def:MP}) is well defined generalization of standard Modus Ponens rule in 2-valued FOL.
\end{coro}
\textbf{Proof}:
The thesis for MV-Modus/Ponens is defined by two sentences $\Gamma = \{\phi_1,\phi_2\}$ where $\phi_1$ is the sentence $\phi$ with the low-boundary truth value $a_1 = f$ and   $\phi_2$ is the sentence $\phi \Rightarrow \psi\}$ with the low-boundary truth value $a_2 = t$, so for each knowledge interpretation $v^* \in \mathbb{V}_\Gamma$ it must hold that  $v^*(\phi) \geq f$ and $v^*(\phi \Rightarrow \psi) \geq t$ (i.e., $v^*(\phi \Rightarrow \psi) =t$ as required by MV-Modus Ponens).
So, from the requirement that $v^*(\phi \Rightarrow \psi) = t$, from the definition of intuitionistic implication in table \ref{tab:impl} it must hold that
\begin{equation} \label{eq:MP}
 v^*(\phi) \leq v^*\psi)
\end{equation}
 For this thesis in $\Gamma = \{\phi,\phi \Rightarrow \psi\} $  the first sentence $\phi$ it can be any value in $X$ and we can use this thesis  with  (\ref{eq:MP}) as satisfaction in (\ref{eq:MVtruth-preserv}) of the modus point entailment $\Gamma \models_v \psi~$, which is satisfied for each knowledge interpretation $v^*$ such that:\\
 Case 1: $v^*(\phi) = f$ and $v^*(\psi) \in \{f,\bot,\top,t\}$, so from false $\phi$ we can entail any sentence as in standard 2-valued FOL Modus Ponens;\\
 Case 2: $v^*(\phi) = \bot$ and $v^*(\psi) \in \{\bot,t\}$;\\
 Case 3: $v^*(\phi) = \top$ and $v^*(\psi) \in \{\top,t\}$, so in this case from inconsistent $\phi$ by MV-Modus Ponens we can entail the inconsistent $\psi$. \\
 Case 4: $v^*(\phi) = t$ and $v^*(\psi) =t$;\\
 Based on these models $v^*$ of $\Gamma$ given in 4 cases above, we have that\emph{ valid entailment} $\phi, \phi \Rightarrow \psi \models \psi$ satisfied for \emph{all these models} of $\Gamma$, would entail $\psi$ iff in all these models $\psi$ is true (in fact the common true-value of $\psi$ in 4 cases above is just the truth-value $t$.\\
   Note that if in the thesis $\Gamma$ we set the low-boundary truth value $a_1 = t$ then we obtain Case 4 above which is standard 2-valued FOL Modus Ponens rule of inference, and hence the MV-Modus Ponens is a conservative many-valued extension of classical 2-valued Modus Ponens. Thus MV-Modus Ponens in (\ref{def:MP}) is well defined.
\\$\square$\\
 \textbf{The autoepistemic robot's reasoning}:

  The introduction of the distinguished predicate $Know$ is fundamental for the conscious  part of strong-AI robot's cognitive system, able to save the robot's learned experience in symbolic way able to support the autoepistemic logic reasoning and deductions.

  The autoepistemic logic is introduced in 2-valued propositional logic \cite{MaTr91} with added universal modal operator, usually written $K$,  and the axioms:
\begin{enumerate}
  \item Reflexive axiom \textbf{T}:   $~~K\phi \Rightarrow \phi$
  \item Positive introspection axiom \textbf{4}:  $~~K\phi\Rightarrow KK\phi$
  \item Distributive axiom \textbf{K}:   $~~(K\phi \wedge K(\phi\Rightarrow\psi)) \Rightarrow K\psi$
\end{enumerate}
for any proposition formulae $\phi$ and $\psi$.\\
However, we do not use more a \emph{pure} logical deduction of the standard 2-valued FOL, but a kind of autoepistemic deduction \cite{Majk04ph,MajkA04} with a proper set of new axioms. The modal (Kripke-like) operator $K$ will be replaced by the IFOL predicate $Know$ applied to the logic formulae transformed into \emph{abstracted terms}.
  Thus, the three epistemic axioms of epistemic modal logic with modal operator $K$,  used to obtain deductive knowledge, can be traduced into 4-valued logic $IFOL_B$ by the following axioms for the predicate $Know$, which are demonstrated to be \emph{the tautologies}\footnote{The fact that the autoepistemic axioms in $IFOL_B$ are just the tautologies (always true) means that indeed natively $IFOL_B$ is an autoepistemic logic to reason about the knowledge in typically human way able to support the Socrate's famous "I know that I do not know that..." But to support the more complex Liar paradox as explained in details in next section.} in $IFOL_B$ which in fact define the semantics of this particular $Know$ predicate, for a given assignment $g:\V\rightarrow \D$ as follows:
\begin{enumerate}
  \item   The modal axiom $\textbf{T}$,  in IFOL is represented by the axiom, for each abstracted term $\lessdot \psi(\textbf{x})\gtrdot^\beta$ in (\ref{eq:esem2}), with $\beta $ the set of variables in the tuple $\textbf{x}$,  by the axiom schema
\begin{equation} \label{eq:asio1}
  Know(t_1,t_2,\lessdot \psi(\textbf{x})\gtrdot^\beta)/g \Rightarrow \psi(\textbf{x})/g
  \end{equation}
  So by using axiom (\ref{eq:asio1}), and FOL deduction, these deductive properties of the robot can deduce  any true single fact (logical sentence) $\psi(\textbf{x})/g$ derived by its neuro-system process, and to render it to robot's consciousness as a single  known fact $Know(t_1,t_2,\lessdot \psi(\textbf{x})\gtrdot^\beta)/g$.
      That is, from (\ref{eq:assAbTerm}) with
    \begin{equation} \label{eq:asio2}
  u_1 =  g^*(\lessdot \psi(\textbf{x})\gtrdot^\beta)= I(\psi(\textbf{x})/g) \in D_1
  \end{equation}
such that $h(u_1) = h(I(\psi(\textbf{x})/g)) \{a\}$ when $\bot \neq a\in X$ (otherwise we set $a =\bot$) is the truth-value of this robot's known sentence $\psi(\textbf{x})/g$ (if this is a Liar sentence (paradox), then $a = \top$ is inconsistent truth-value, and robot is able to work also with inconsistent knowledge as well.  So, this $\textbf{T}$  axiom (\ref{eq:asio1}) becomes
\begin{equation} \label{eq:asio1b}
  Know(g^*(t_1),g^*(t_2),u_1) \Rightarrow \psi(\textbf{x})/g
    \end{equation}
Note that the meaning of the intensional concept $u_1$ of the robot is grounded on robot's neuro-system process, which is just robot's current internal experience of what is he doing. For each knowledge model $v^*$ the truth-values of both sides of implication in (\ref{eq:asio1}) are equal, so the truth-value of this logic schema (\ref{eq:asio1}) is true (lies in the diagonal of the truth-table (\ref{tab:impl})), that is, it is really an reflexive axiom.\\
Consequently, the application of the \textbf{T} axiom  allows the extraction from robot's conscious knowledge the \emph{logical sentences} which, successively, can be elaborated by robot's implemented deductive property of FOL in two ways:\\

a.1. To respond to some human natural language questions (parsed into a logical formula) and to verify if the response  is "yes" or "no", or "I do not know" (if robot's conscious knowledge is incomplete for such a question);\\

a.2. To deduce another sentences which then can be inserted in robot's  conscious knowledge as ground atoms of the predicate $Know$ (where this deduced sentence is represented as an abstracted term). This process (in background, or when robot is free of other concrete activities) can be considered as a kind of consolidation and completion of robot's knowledge based on previous experiences, in an analog way as it is done by human mind when we sleep.\\
\item The positive introspection axiom \textbf{4}:
\begin{equation} \label{eq:asio1b4}
 Know(t_1,t_2,\lessdot \psi(\textbf{x})\gtrdot^\beta)/g \Rightarrow Know(g^*(t_1),g^*(t_2),\lessdot Know(t_1,t_2,\lessdot \psi(\textbf{x})\gtrdot^\beta)/g\gtrdot)
 \end{equation}
 that is.
 \begin{multline} \label{eq:asio1b4a}
 Know(g^*(t_1),g^*(t_2),g^*(\lessdot \psi(\textbf{x})\gtrdot^\beta)) \Rightarrow\\ Know(g^*(t_1),g^*(t_2),\lessdot Know(g^*(t_1),g^*(t_2),g^*(\lessdot \psi(\textbf{x})\gtrdot^\beta))\gtrdot)
  \end{multline}
 which, in the case when $g(y_1) = in ~present$ and $g(y_2) = me$, is traduced in natural language by robot as:\\
 "I know that $\psi(\textbf{x})/g$" implies "I know that I know that $\psi(\textbf{x})/g$"\\
 where $g^*(\lessdot \psi(\textbf{x})\gtrdot^\beta) = I(\psi(\textbf{x})/g)$ contains the logic truth-value of the known sentence $\psi(\textbf{x})/g$ as explained by (\ref{eq:asio2}).\\
 For each knowledge model $v^*$ the truth-values of both sides of implication in (\ref{eq:asio1b4a}) are equal, so the truth-value of this logic schema (\ref{eq:asio1b4}) is true (lies in the diagonal of the truth-table (\ref{tab:impl})), that is, it is really an introspection axiom schema.\\
  \item The distributive axiom \textbf{K} ("modal Modus Ponens"):
  \begin{multline} \label{eq:asio1b4}
 (Know(t_1,t_2,\lessdot \psi(\textbf{x})\gtrdot^\beta)/g \wedge Know(t_1,t_2,\lessdot \psi(\textbf{x}) \Rightarrow \phi(\textbf{z})\gtrdot^{\beta\bigcup\beta_1})/g) \\
 \Rightarrow Know(t_1,t_2,\lessdot \phi(\textbf{z})\gtrdot^{\beta_1})/g
    \end{multline}
    with the \emph{sets} of variables  $ \beta = \overline{\textbf{x}}$ and $\beta_1 = \overline{\textbf{z}}$. Or, equivalently,
  \begin{multline} \label{eq:asio1b4c}
 (Know(g^*(t_1),g^*(t_2),g^*(\lessdot \psi(\textbf{x})\gtrdot^{\beta})) ~\wedge\\ Know(g^*(t_1),g^*(t_2),g^*(\lessdot \psi(\textbf{x}) \Rightarrow \phi(\textbf{z})\gtrdot^{\beta\bigcup\beta_1}))) \\
 \Rightarrow Know(g^*(t_1),g^*(t_2),g^*(\lessdot \phi(\textbf{z})\gtrdot^{\beta_1}))
    \end{multline}
 Note that this axiom schema is a way how the robot provides the conscious implications $\psi(\textbf{x})/g \Rightarrow \phi(\textbf{z})/g$,
 independently if they are its innate implemented rules (introduced in robot's knowledge when it is created) or if they are learned by robot's own experience.  For any knowledge model $v^*$, if we apply it to the logic formula (\ref{eq:asio1b4c}) by using (\ref{eq:esem3s}), from point 2 of Definition \ref{def:MV-algebra} for binary operators $\wedge$ and $\Rightarrow$, we obtain the propositional formula with truth-values $a = v^*(\psi(\textbf{x})/g)\in X$ and $b = v^*(\phi(\textbf{z})/g)\in X$,
 $$(a \wedge (a \Rightarrow b)) \Rightarrow b$$
 which is true for all cases of $a, b \in X$. Thus (\ref{eq:asio1b4}) is a valid axiom schema.
\end{enumerate}
\textbf{Remark}: We recall that this method of application of autoepistemic deduction  (for concepts such as \emph{knowledge}) can be applied to all other modal logic operators (for concepts such as  \emph{belief}, \emph{obligation}, \emph{causation}, \emph{hopes}, \emph{desires}, etc., for example by using \emph{deontic} modal logic that same statement have to represent a \emph{moral obligation} for robots), by introducing special predicates  for them with the proper set of axioms for their active semantics (fixing their meaning and deductive usage).

  By such fixing by humans of robot's unconscious  part with active semantics (which can not be modified by robots and their live experience) of all significant for human robot's concepts and their properties, we will obtain ethically confident and socially safe and non danger robots (controlled by public human  ethical security organizations for the production of robots with general strong-AI capabilities).
 \\$\square$\\
 Note that the obtained robot's knowledge,  from the known facts  at the end of deduction, is in robot's temporary memory. In order to render it permanent (by cyclic process of transformation of the temporary into permanent robot's memory), we need to add to any predicate of the robot's FOL syntax, also the time-variable as, for example, the first variable of each predicate (different from $Know$), instantiated in the known facts by the  \emph{tamestamp} value $\tau$ (date/time) when this knowledge of robot is transferred into permanent memory.
This temporization of all predicates used in robot's knowledge is useful for robot to search all known facts in its permanent memory that are inside some time-interval as well.

 It can be used not only to answer directly to some human questions about robot's knowledge, but also to extract only a part of robot's knowledge from its permanent memory in order to be used for robot's deduction, and hence to answer to human more complex questions that require deduction of new facts not already deposited in explicit robot's known facts.

 The last argument of this section will be dedicated to the evolution of robot's knowledge in time.\\
\textbf{The problems about knowledge fixpoints}:

 Based on  the bilattice $(X,\leq,\leq_k)$ of truth-values in $X$ we are able to define the truth and knowledge partial orders
  over the knowledge interpretations in $\mathbb{V}$ as well:
  \begin{definition} \label{def:BIL}
    Let $H_0 \subset H$ be the subset of ground atoms of predicates in $P$ that are not built-in or distinguished predicates.  Then for $v^*_1, v^*_2 \in \mathbb{V}$ we have that\\

1. Truth partial order $\preccurlyeq$ is defined by

$v^*_1 \preccurlyeq v^*_2$ ~~~~iff~~~~
$(\forall A \in H_0). v_1(A) \leq v_2(A)$;\\

2. Knowledge partial order $\preccurlyeq_k$ is defined by

$v^*_1 \preccurlyeq_k v^*_2$~~~~
iff ~~~~$(\forall A \in H_0). v_1(A) \leq_k v_2(A)$.
\end{definition}
Note that the truth-value of ground atoms of built-in predicates are predefined initially and do not change in time, while the ground atoms of distinguished predicates $T$ and $Know$ are derived from the truth-values of ground atoms that compose the sentence in their abstracted terms. Because of that we used only ground atoms in $H_0$ for definition of two partial orders in definition above.

In particular, if we are considering the growing in time of the robots knowledge, with his knowledge interpretation $v_i$ at time instance $\tau_i$, and  the knowledge interpretation $v_{i+1}$ at time instance $\tau_{i+1}$, with $\tau_i < \tau_{i+1}$.

Let us denote by $v^*_0$ the bottom knowledge interpretation in which each ground atom $A \in H_0$ (of non-distinguished predicates) is unknown. So,
\begin{enumerate}
\item  The robot's learning process which increments its knowledge is obtained by passing some ground atom $A$ (of non-distinguished and non-built-in predicates) from the initial truth-value $\bot$ (unknown), to the learned by robot truth-value $f$ (false) or $t$ (true). An  atom is not learned by a robot  to be unknown, but only initially (at the beginning when a  robot is created) for default it is set to be unknown in robot's initial (bottom) knowledge interpretation $v^*_0$. Thus one time robot learned that an atom is true or false, successively can not learn that it is unknown.\\It is possible that such an atom that robot previously learned to be false, in future by different robot's experience can become inconsistent (if robot learned successively from different sources that it could be true). Analogously, an atom previously learned to be true, in future can  become inconsistent.
\item  If an ground atom becomes inconsistent its truth-value can not be successively changed, like in the case of built-in paradoxes\footnote{In effect, all human paradoxes are initially introduced in robot's knowledge as complex sentences with requirement that, in each knowledge interpretation of them, have the same constant built-in truth-value $\top$}.
\end{enumerate}
Based on the considerations above and Definition \ref{def:BIL}, we obtain the following result:
\begin{lemma} \label{lemma:monK}
 We have that during time-evolution for $\tau_i \leq \tau_{i+1}$ of robot's knowledge interpretations  it is always satisfied that
 \begin{equation} \label{eq:knowledge}
 v^*_i\preccurlyeq_k v^*_{i+1}
\end{equation}
 that is, we have a monotonic increments of knowledge.
\end{lemma}
\textbf{Proof}: We have that for each ground atom $A \in H_0$ its truth-value can evolve only in this way:\\
1. $\bot \mapsto f$  of $\bot \mapsto t$;\\
2. $f \mapsto \top$;\\
3. $t \mapsto \top$;\\
4. $\top \mapsto \top$.\\
Thus, all changes of truth-values of ground atoms in $H_0$ follow the increments of knowledge in the lattice of knowledge ordering of Belnap's bilattice. Thus, form Definition \ref{def:BIL} the relationship (\ref{eq:knowledge}) is valid.
\\$\square$\\
Consequently, given such time-ordering of extended robot's knowledge, expressed by $v^*_i$ in the ordered sequence $v^*_0 \prec_k...\prec_k v^*_i\prec_k v^*_{i+1}\prec_k ...$, based on Lemma \ref{lemma:monK}, we are able to define the monotonic function $f_k:\mathbb{V}\rightarrow \mathbb{V}$ such that for $i = 1,2,3,...$
\begin{equation} \label{eq:monK}
 f_k(v^*_i) =  v^*_{i+1}
\end{equation}
In fact, for any two $v^*_j, v^*_n$ such that $v^*_j\preccurlyeq_k v^*_n$ we have the monotonic property $f_k(v^*_j)\preccurlyeq f_k(v^*_n)$.
Note that we do not guarantee the monotonicity w.r.t the truth-ordering in time evolution of robots knowledge. So, we obtain the following result:
\begin{coro} \label{coro:monK}
In any ordered sequence $v^*_0 \prec_k...\prec_k v^*_i$  with $S = \{A\in H_0 |v_i(A) =\bot\}$  if $S$ is empty then $v^*_i$ can be the fixed point of $f_k$, that is, $f_k(v^*_i) = v^*_i$.
\end{coro}
\textbf{Proof}: If $S$ is empty than if all inconsistencies are found, $v^*_i$ is fixed point, that is,  successive changes of robot's knowledge are impossible, that is, it will not be possible for a robot to learn more by its existing logic system.
\\$\square$\\
Of course, if robot reaches a knowledge fixpoint  by learning process robot can not increment its reached knowledge. \\
So, in order to avoid knowledge fixpoints,in order to  be able to increment its knowledge,
the AGI robots must be able to learn and generate (like the humans) from the natural language an extra new predicate $p^k_j \notin P$ (w.r.t the initially defined set $P$) whose atom with all arguments the sorted variables in k-ary tuple of variables $\textbf{x}$, $p^k_j(\textbf{x})$ and its corresponding PRP concept (real or abstract) $u = I(p^k_j(\textbf{x}))$ where $I$ is new intensional  interpretation  by initial assignment to all ground atoms of this new predicates the logic truth-value $\bot$. That with empty set $h(u)$ of extension of this new concept.
\section{Paradoxes: Resolution of Liar Formula in $\textit{IFOL}_B$  \label{section:Liar}}
Let us introduce the Liar paradox in the standard 2-valued FOL of
arithmetic (standard \emph{2-valued} first-order Peano Arithmetic),
extended by a primitive satisfaction binary predicate
(G$\ddot{o}$del) $Sat(x,y)$, governed by Tarskian axioms. This
system of axioms governing $Sat(x,y)$ was given in \cite{Tar1936}.

The philosophical significance of such FOL logic is discussed in
\cite{Ketl99}, compared with minimalistic extensions generated by
adding just the "\textbf{T}-schema" (the set of formulae \index{Liar formula resolution}
$\textbf{T}(\ulcorner\phi\urcorner) \Leftrightarrow \phi$ where
$\ulcorner\phi\urcorner$ is the code (positive integer) obtained by
G$\ddot{o}$del's codification of logic formula $\phi$, and $\textbf{T}$ is
unary \emph{2-valued} truth-predicate whose argument has the sort of natural numbers $N$.
Similar constructions are considered in \cite{Shap98}.

 The idea is then $Sat(x,y)$ expresses the satisfaction relation between (codes of) formulae and (codes of) sequences of the formulae. Then $\textbf{T}(x)$ expresses the \emph{concept of truth} for such formulae and sequences. In fact, it can be shown
that this theory satisfies Tarski's Convention (T) in \cite{Majk22}: i.e., that
\begin{equation} \label{eq:Liar1}
 \textbf{T}(\ulcorner\phi\urcorner) \Leftrightarrow \phi~~~~~
\end{equation}
 is a theorem (true formula) in this 2-valued Peano Arithmetic FOL.

If $\phi$ is any formula with one free variable $x$, then the
\emph{diagonalization} of $\phi$, can be defined by mapping $d:\L
\rightarrow \L$, such that $d(\phi(x)) = \phi(\ulcorner
\phi(x)\urcorner)$, equivalent to the formula $(\exists y)( y =
\ulcorner \phi(x)\urcorner \wedge \phi(y))$.\\
Based on this mapping we can define the mapping $diag:N \rightarrow
N$, such that for any non-negative integer $n \in N$, $diag(n) =
\ulcorner d(\phi(x))\urcorner$ if $n = \ulcorner \phi(x)\urcorner$;
$0$ otherwise.

Now we can introduce the binary predicate $Diag(x,y)$ such that for
any two $n,k \in N$, $Diag(k,n)$ is true iff $n = diag(k) \neq 0$,
that is, if $n$ is the code of the formula obtained by diagonalization of the formula $\phi(x)$ whose code is equal to $k$.

From the Diagonalization Lemma (or Fixed Point Theorem) in Peano
Arithmetic FOL, there must be fixed point formula $\lambda$ such
that, the following formula is a theorem (true formula) in this logic:
\begin{multline} \label{eq:Liar2}
 ~~~~~~~~~~~~~\lambda \Leftrightarrow ~ \neg \textbf{T}(\ulcorner \lambda
\urcorner), ~~~~ $for$ ~\lambda ~$equal to$ ~d(\neg \textbf{T}(diag(x)))~~~~~~~~~~~~~~~~~
\end{multline}
It is easy to show that $\neg \textbf{T}(diag(x))$ is equal to the formula
 $(\exists y)( Diag(x,y) \wedge \neg \textbf{T}(y))$ with the free variable
 $x$.
 Thus,
\begin{multline} \label{eq:Liar3}
$Let~$ \phi(x) $ ~be ~equal~ to formula~$ (\exists y)( Diag(x,y) \wedge \neg \textbf{T}(y)), $~then$\\
  ~~~~\lambda $ ~is equal to the formula ~$ d(\phi(x)), $i.e, to$~(\exists z) (Diag((\ulcorner \phi(x)\urcorner,z) \wedge \neg  \textbf{T}(z))\\
  $and hence, to$~~
 (\exists z) (Diag((\ulcorner \exists y)( Diag(x,y) \wedge \neg
 \textbf{T}(y))\urcorner,z) \wedge \neg  \textbf{T}(z))~~~
\end{multline}
The analysis of the proof of this Diagonalization Lemma
\cite{BoJe98,Ketl00} shows that this formula $\lambda$ must contain
the atoms $\textbf{T}(x)$ of 2-valued  predicate $\textbf{T}$. This formula is the formal analogue of the
so-called "\emph{strengthened liar}" for this system. It is a formula that
"\emph{says of itself that it is not true}".

In fact, it expresses the Liar paradox. If we assume that $\lambda$
is true than $\neg \textbf{T}(\ulcorner \lambda \urcorner)$ must be true,
i.e. $ \textbf{T}(\ulcorner \lambda \urcorner)$ must be false and, consequently, from (\ref{eq:Liar1}) we obtain that $\lambda$ is false: the contradiction. The same contradiction we obtain if we assume that $\lambda$ is false.
\\\\
\textbf{Liar paradox in 4-valued} $\L_{in}$:\\
 Let us consider now this problem inside the intensional many-valued
FOL $\L_{in}$ defined previously by having also the unary predicate $\textbf{T}$ with sort of natural numbers the binary predicate \emph{Diag} with both arguments of the sort of natural numbers, so that the formula $\lambda$ in (\ref{eq:Liar3})  is well defined in 4-valued intensional FOL $\L_{in}$ as well.
However, while \emph{Diag} in $\L_{in}$ remains 2-valued (can be only true or false), the predicate $\textbf{T}$  in $\L_{in}$ is 4-valued and hence can have also the truth-values $\bot$ and $\top$ as well.  Thus also the formula $\neg \textbf{T}(\ulcorner \lambda
\urcorner)$  in (\ref{eq:Liar2}) can have the truth-value $\bot$ or $\top$ as well
and hence the formula (\ref{eq:Liar2}) can remain a theorem (true formula) also if $\lambda$ has the truth-value $\bot$ or $\top$. Let us show this fact:
from (\ref{eq:Liar3}) and knowledge interpretation $v^*$, we obtain

$v^*(\lambda) = v^*((\exists z) (Diag((\ulcorner \phi(x)\urcorner,z) \wedge \neg  \textbf{T}(z)))$\\  and hence from the point 3 of Definition \ref{def:MV-algebra} for existentially quantified formulae, and the fact that for the ground atoms  $v^*(Diag((\ulcorner \phi(x)\urcorner,z)/g_1)$  are false for all $g_1(z) \neq \ulcorner \phi(\ulcorner \phi(x)\urcorner)\urcorner \in N$, we obtain that,

$v^*(\lambda) = v^*(\neg\textbf{T}(\ulcorner \phi(\ulcorner \phi(x)\urcorner)\urcorner)) = \neg v^*(\textbf{T}(\ulcorner \phi(\ulcorner \phi(x)\urcorner)\urcorner)) \in X$\\
can have any truth-value in $X$, thus also $\bot$ or $\top$.

Consequently, instead of the formula
 $~\lambda \Leftrightarrow ~ \neg \textbf{T}(\ulcorner \lambda \urcorner)$ in (\ref{eq:Liar2}) for the sentence $\lambda$ defined in (\ref{eq:Liar3}), we can use the intensional abstraction operator  $\lessdot \_ \gtrdot$  and the 4-valued unary predicate $T$ which argument is the sort of abstract terms and obtain equivalent to it formula
\begin{equation} \label{eq:Liar2A}
 \lambda \Leftrightarrow ~ \neg T( \lessdot \lambda \gtrdot)
\end{equation}
Note that both unary predicates $\textbf{T}$ and $T$ have the same semantics as truth-predicates but only represent the formula $\lambda$ in two different ways: in $\textbf{T}$ is represented by its numeric G$\ddot{o}$del's code while in $T$ by its abstracted term used in intensional FOL.
However, in this case we have that the unary predicate $T$  is \emph{many-valued}
  and hence able to satisfy the condition (\ref{eq:Liar2A}) by using the truth-constraint provided in (\ref{eq:esem3s})
 \begin{equation} \label{eq:Tmv}
  v^*(T( \lessdot \phi(\textbf{x})/g \gtrdot)) = v^*(\phi(\textbf{x})/g) \in X
\end{equation}
for any sentence $\phi(\textbf{x})/g$ and valuation $v^*$.
Thus,  for each MV-interpretation $I^*_{B}$ of $\L_{in}$, we may fix $I^*_{B}(\lambda) = \{\kappa\} \notin \{\{f\}, \{t\}\}$ in the way that the  Liar sentence $\lambda$ satisfy (\ref{eq:Liar2A}),  with the many-valued unary predicate $T$:
\begin{coro} \label{coro:Liar} The Liar self-reference sentence is not a paradox in $\L_{in}$, thus, does not create the contradictions.
\end{coro}
\textbf{Proof:} Note that from Definition 8 in \cite{Majk25}, for a relation $R\in \mathfrak{Rm}$ with arity $ar(R) =1$, for the unary operator $~\oslash:\mathfrak{Rm}\rightarrow \mathfrak{Rm}$,
 $~\oslash(R) =  \{\neg a  ~|~ \{a \} = R, ~\neg a \neq \bot \} $.

Thus, for any  MV-interpretation $I^*_{B}$, such that
$I^*_{B}(\lambda) = \{\kappa\}$ (or an equivalent many-valued
interpretation $v^*$ such that $v^*(\lambda) = \kappa$), such that $\neg \kappa = \kappa$, which is satisfied for $\kappa =\top$ ("inconsistent" truth-value, both true and false) in the lattice of truth values $X$,  we have that $I^*_{B}(\neg T(\lessdot \lambda \gtrdot)) = \oslash I^*_{B}(T(\lessdot \lambda \gtrdot)) =$ from (\ref{eq:Tmv}) $=\oslash (\{\top\}) =  \{\top\} = I^*_{B}(\lambda)$ so that (\ref{eq:Liar2A}) (that is, the (\ref{eq:Liar2})) are true formulae. Thus, in $\L_{in}$ the  value to $\lambda$ is assigned without any contradiction.
\\$\square$\\
Analogously, the G$\ddot{o}$del's contradictory formulae obtained by
diagonalization of the formula $\neg Prov(diag(x))$, where $Prov(x)$
is a formula representing provability in the Peano arithmetic FOL can be resolved in $\L_{in}$ by interpreting these formulae with the  value $\kappa = \bot$ ("unknown" truth-value) of Belnap's bilattice. \index{Belnap's bilattice}

\section{Conclusions and Future Work}

By proposed four-level cognitive robot's structure, IFOL allow robots to reflect on, reason about, anticipate, or simply imagine scenes, situations, and developments within in a highly flexible, compositional, that is, semantically meaningful manner. As a result, IFOL enables the robots to actively infer highly flexible and adaptive goal-directed behavior under varying circumstances \cite{Russ20}.\\

\emph{What psychological and philosophical significance should we attach to recent efforts at computer simulations of
human cognitive capacities? In answering this question, I find it useful to distinguish what I will call "strong" AI
from "weak" or "cautious" AI (Artificial Intelligence). According to weak AI, the principal value of the computer
in the study of the mind is that it gives us a very powerful tool. For example, it enables us to formulate and test
hypotheses in a more rigorous and precise fashion. \\But according to strong AI, the computer is not merely a
tool in the study of the mind; rather, the appropriately programmed computer really is a mind, in the sense that
computers given the right programs can be literally said to understand and have other cognitive states. In strong
AI, because the programmed computer has cognitive states, the programs are not mere tools that enable us to
test psychological explanations; rather, the programs are themselves the explanations.}  In Minds, brains, and programs, John R.Searle, \cite{Sear80}\\

So, in our proposal, the \emph{cognitive states} of AGI robots are defined by extensions of its knowledge Database for intensional concepts and known sentences represented by the ground atoms of the $Know$ predicate in the actual knowledge model. So, mental content od an AGI robot is represented by its memory as a conscious cognitive state and by the grounding of the intensional language concepts to all kind of robot's processes realized by algorithms, neural networks and other computational and programming methods like subsumption architecture  of Rodney Brooks  \cite{Broo99} based on layers and augmented finite-state machines.

IFOL-based approach is part of a general neuro-symbolic AI paradigm, where researchers aim to blend: Neural methods (e.g., deep learning) for pattern recognition and learning from data, and symbolic logic for structured knowledge representation and reasoning. This broader field (which includes work at IBM Research, MIT, and in academic surveys) seeks to build AI systems that can both \emph{learn from experience} and \emph{reason abstractly}. These ideas  contribute to ongoing discussions in AI about \emph{symbol grounding}, \emph{logical inference}, and \emph{self-referential reasoning} — all important for Strong AI research.

The future work will be to extend the $IFOL_B$ to probabilistic reasoning and to consider an implementation of the Symbol Grounding Problem
(SGP) for the symbolic system of $IFOL_B$. For example, eight proposed strategies for solving the SGP, which was given its classic formulation in Harnad are provided in \cite{MaFl05}.
The SGP concerns the possibility of specifying precisely how a
robot can autonomously elaborate its own semantics for the symbols that it manipulates and do so from scratch, by interacting with its environment and other autonomous agents. This means that, as Harnad rightly emphasises \cite{Harn90}, the interpretation of the symbols must be intrinsic to the symbol system itself, it cannot be extrinsic. This is provided in $IFOL_B$ by the fact that this logic language support both the truth and the meaning by the PRP concepts.
\newpage

\newpage
\section{Appendix: Neuro-symbolic Archetuture}
. \\\\\\
 \begin{figure}
$\vspace*{-44mm}$
\centering{
 \includegraphics[scale= 1.0]{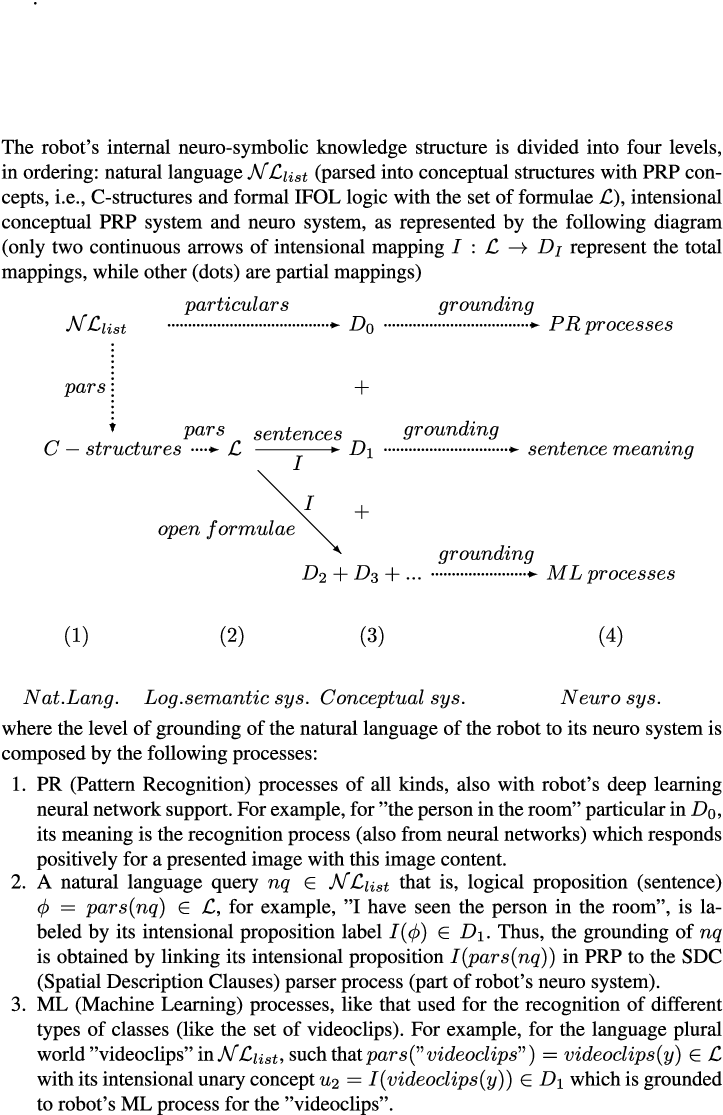} }
  $\vspace*{-11mm}$
 \end{figure}
\newpage
In this research, bridging the gap between Neural Networks (sub-symbolic) and Logic (symbolic) is not just about making them work side-by-side; it is about creating a formal translation layer where the robot's "internal feelings" (neuron firings) become "meaningful concepts" (logical terms).
We achieve this via a Dual-System Architecture grounded in Intensional First-Order Logic (IFOL):
\begin{enumerate}
  \item \textbf{System 1: The "Neuronal Grounding" Layer};\\
System 1 consists of Deep Learning and Neural Networks. For us, this layer is responsible for pattern recognition and sensory-motor coordination.\\
- \textbf{The Output}: Instead of just outputting a label like "Blue," the neural network produces a specific internal state (a vector or firing pattern).\\
- \textbf{The Problem}: On its own, the neural network doesn't "understand" the concept; it just reacts. This is where the gap exists.
\\
  \item \textbf{The Bridge: Intensional Abstraction}:\\
This is our "secret sauce." We use intensional abstraction to turn the complex activity of the neural network into a "Logic Object."\\
- \textbf{Mapping}: Each distinct neural pattern is mapped to a specific intensional entity in the IFOL.\\
- \textbf{Meaning vs. Reference}: The "meaning" (intension) of a word like heavy is the specific neural experience of the robot's motors straining. The "reference" (extension) is the actual physical object being lifted.\\
- \textbf{Self-Labeling}: The robot uses autoepistemic reasoning to say, "I am currently experiencing Neural State X, which I have mapped to the concept 'Heavy'."
\\
  \item \textbf{System 2: The "Autoepistemic Logic" Layer}:\\
Once the neural patterns are abstracted into logical symbols, System 2 takes over. This is the deliberative, slow-thinking part of the brain.\\
- \textbf{High-Level Planning}: Because these symbols are now part of a formal logic (IFOL), the robot can use them in complex "If-Then" scenarios that neural networks struggle with.\\
- \textbf{Predefined set of axioms and prohibitions}: Any kind of reasoning results and plans must satisfy these axioms before the executions of actions derived from such plans. In this way the robots can not become danger for the humans and the environment in which they are active.\\
- \textbf{Feedback Loop}: System 2 can "query" System 1. For example, the symbolic layer might ask, "Check the visual sensors again; does that object match the 'Bird' intension?" This forces the neural network to re-process data based on logical needs.
\\
\item \textbf{Solving the "Black Box" Problem}:\\
One of the most interesting aspects of this bridge is Explainability.\\
- In standard neural networks, we don't know why a robot chose an action.\\
- In this framework, because every neural state is linked to an intensional logic term, the robot can provide a symbolic trace of its reasoning: "I performed Action A because my neural sensors triggered the 'Danger' intension, which my logic defines as a state to be avoided".
\end{enumerate}
\textbf{Summary Table:} The neuro-symbolic's Bridge\\
\begin{tabular}{|c|c|c|c|}
  \hline
    & \textbf{Sistem 1 (Neural)} & \textbf{The Bridge (Intensional)} & \textbf{System 2 (Symbolic)} \\
    \hline
  \textbf{Function} & 	Fast, reactive sensing & Abstraction $\&$ Grounding & Slow, logical planning \\
  \hline
  \textbf{Data Type} & Vectors / Tensors & Intensional Entities & Logic Formulas \\
  \hline
  \textbf{Role} &Perceives the world & Maps neurons to symbols & Reasons about perceptions \\
  \hline
  \textbf{Self-awareness} & Unconscious & The "\textbf{I}" links the two & Conscious autoepistemic \\
  & & & thought \\
  \hline
\end{tabular}
\end{document}